\pdfoutput=1
\documentclass[12pt]{article}
\usepackage{mystyle}
\usepackage{microtype}
\usepackage{graphicx}
\usepackage{subfigure}
\usepackage{booktabs} % for professional tables
\usepackage{fullpage}
\usepackage{natbib}
\usepackage{hyperref}
\usepackage{tikz}
\usepackage{acronym}
\usetikzlibrary{positioning, shapes.geometric}
\usepackage{geometry}

\acrodef{marl}[MARL]{Multi-Agent Reinforcement Learning}
\acrodef{mfg}[MFG]{Mean-Field Game}
\acrodef{mfc}[MFC]{Mean-Field Control}
\acrodef{gmfg}[GMFG]{Graphon Mean-Field Game}
\acrodef{mdp}[MDP]{Markov Decision Process}
\acrodef{sbm}[SBM]{Stochastic Block Model}
\acrodef{ne}[NE]{Nash Equilibrium}
\acrodef{rkhs}[RKHS]{Reproducing Kernel Hilbert Space}
\acrodef{gmfgppo}[GMFG-PPO]{Proximal Policy Optimization for GMFG}
\acrodef{ppo}[PPO]{Proximal Policy Optimization}
\acrodef{sbm}[SBM]{Stochastic Block Model}
\acrodef{mgmfgpmd}[MonoGMFG-PMD]{Monotone GMFG Policy Mirror Descent}
\title{Learning Regularized Monotone Graphon Mean-Field Games}

\author{Fengzhuo Zhang$^{1}$ \quad Vincent Y.~F.~Tan$^1$ \quad Zhaoran Wang$^2$\quad Zhuoran Yang$^3$\\
$^1$National University of Singapore \quad $^2$ Northwestern University   \quad $^3$Yale University\\
\texttt{fzzhang@u.nus.edu}, \texttt{vtan@nus.edu.sg}, \\
\texttt{zhaoranwang@gmail.com}, \texttt{zhuoranyang.work@gmail.com}
}

\begin{document}

\maketitle
\begin{abstract}
    This paper studies two fundamental problems in regularized Graphon Mean-Field Games (GMFGs). First, we establish the existence of a Nash Equilibrium (NE)  of any $\lambda$-regularized GMFG   (for  $\lambda\geq 0$). This result relies on weaker conditions than those in previous works for analyzing both unregularized GMFGs ($\lambda=0$) and $\lambda$-regularized MFGs, which are special cases of GMFGs. Second, we propose provably efficient algorithms to learn the NE in  weakly monotone GMFGs, motivated by \cite{lasry2007mean}. Previous literature either only analyzed continuous-time algorithms or required extra conditions to analyze discrete-time algorithms. In contrast, we design a discrete-time algorithm and derive its convergence rate  solely under  weakly monotone conditions. Furthermore, we develop and analyze the action-value function estimation procedure during the online learning process, which is absent from algorithms for monotone GMFGs. This serves as a sub-module in our optimization algorithm. The efficiency of the designed algorithm is corroborated by empirical evaluations.
\end{abstract}

\section{Introduction}
In \ac{marl}, the sizes of state and action spaces grow exponentially in the number of agents,   which is known as the ``curse of many agents''~\citep{sonu2017decision,wang2020breaking} and restrict its applicability to  large-scale scenarios. The \ac{mfg} has thus been proposed to mitigate this problem by exploiting the  homogeneity assumption~\citep{huang2006large,lasry2007mean}, and it has achieved tremendous successes in many real-world applications~\citep{cousin2011mean,achdou2020mean}. However, the homogeneity assumption is an impediment when modeling  scenarios in which the agents are heterogeneous. GMFGs, as extensions of \ac{mfg}s, have thus been proposed to model the behaviors of heterogeneous agents and ameliorate the ``curse of many agents'' at the same time~\citep{parise2019graphon,carmona2022stochastic}. 

Despite the empirical successes of the \ac{gmfg}~\citep{aurell2022finite}, its theoretical understanding   is  lacking. First,  sufficient conditions for \ac{ne} existence in  regularized \ac{gmfg} have not been established. Most works only address the existence of the \ac{ne} in {\em unregularized} \ac{gmfg}s. However,  {\em  regularization} is employed in practical implementations for  improved exploration and robustness~\citep{geist2019theory}. Moreover, previous works prove the existence of \ac{ne} in regularized \ac{mfg}s, a special case of \ac{gmfg}s, only under the {\em contraction condition}, which is overly restrictive  for real-world applications. Second, the analysis of  {\em discrete-time} algorithms for monotone \ac{gmfg}s is lacking.  Most existing works design  provably efficient discrete-time algorithms only under contraction conditions, as shown in Table~\ref{table:compare}. Complementarily, previous works on monotone \ac{gmfg}s either only derive the convergence rate for continuous-time algorithms, which ignores the discretization error, or require extra conditions, (e.g., potential games) to analyze discrete-time algorithms.

In this paper, we first consider \ac{gmfg}s in   full generality, i.e., without any contraction or monotone conditions. The goal is to establish the existence of \ac{ne} in the regularized \ac{gmfg} in this general setting. Then we focus on  monotone \ac{gmfg}s motivated by \citet{lasry2007mean}. We aim to learn the unique \ac{ne} from the online interactions of all agents with and without the action-value function oracle. When the oracle  is absent, the action-value functions should be estimated from the data of sampled agents generated in the online game.

In the analysis, we have to overcome  difficulties that arise from both the {\em existence} of the \ac{ne} problem and the {\em learning} of the  \ac{ne}. First, the proof of the existence of \ac{ne} in regularized \ac{gmfg} involves establishing  some topological spaces and operators related to \ac{ne} on which fixed point theorems are applicable. However, the direct construction of the space and the operators for \ac{gmfg}s with uncountably infinite agents is challenging. Second, the design and analysis of the discrete-time \ac{ne} learning algorithm require subtle exploitation of the monotone condition. Unlike  continuous-time algorithms with infinitesimal step sizes, the design of appropriate step sizes is additionally required for the discrete-time algorithm to guarantee that iterates evolve appropriately. This guarantee originates from the delicate utilization of the monotone condition in the optimization procedures.

To address the difficulty of the existence problem, we construct a regularized MFG from the regularized \ac{gmfg} and show that the \ac{ne} of the constructed game can be converted into the \ac{ne} of the original game, thus  mitigating the difficulty of having an uncountable number of agents. To handle the difficulty in the \ac{ne} learning problem, we design the \ac{mgmfgpmd} algorithm, which iteratively implements policy mirror descent for each player. We show that this procedure results in a decrease of the KL divergence between the iterate and the \ac{ne}, and this decrease is related to the gap appearing in the weakly monotone condition. When the action-value function oracle is absent, we also design and analyze action-value functions estimation algorithms to serve as a submodule of the optimization procedures.

\underline{{\bf Main Contributions:}} We first establish the existence of the \ac{ne} in the $\lambda$-regularized \ac{gmfg} with $\lambda\geq 0$ assuming  Lipschitzness of graphons and  continuity of transition kernels and reward functions. Our result relaxes the assumption of the  Lipschitzness of transition kernels and rewards required in  previous works on unregularized \ac{gmfg}s and the contraction condition in the literature on regularized \ac{mfg}~\citep{cui2021approximately}. Then we design and analyze the \ac{mgmfgpmd} algorithm. Using an action-value function oracle, the convergence rate for \ac{mgmfgpmd} is proved to be  $\tilO(T^{-1/2})$ after $T$ iterations. Without the oracle, the convergence rate includes  an additional $\tilO(K^{-1/2}+N^{-1})$  term  that arises from sampling $N$ agents and collecting data from $K$ episodes, reflecting the generalization error and the approximation error of the  estimation algorithm. As shown in Table~\ref{table:compare} , our algorithm can be implemented from the online interaction of agents and does not require the distribution flow manipulation. Detailed explanations of the properties stated in the columns of Table~\ref{table:compare} are provided 
in Appendix~\ref{app:compare}. Our result for \ac{mgmfgpmd} provides the {\em first convergence rate for  discrete-time algorithms in monotone \ac{gmfg}s}.

\begin{table}[t]
	\centering
	\caption{Comparison of \ac{gmfg}s and \ac{mfg}s learning algorithms}
	\setlength\tabcolsep{2pt}
	\begin{tabular}{ccccccc}
		\hline
		 & Condition & \thead{No population\\manipulation} & \thead{Online\\playing} & Heterogeneity & \thead{Discrete-time\\ algorithm}&\thead{Convergence\\rate} \\ 
		\hline
		\citet{anahtarci2022q}& Contraction & No & No & No &Yes & Yes\\
		\citet{xie2021learning}& Contraction & No & Yes & No &Yes & Yes\\
		\citet{zaman2022oracle}& Contraction  & No & Yes & No &Yes & Yes\\
		\citet{yardim2022policy}& Contraction & Yes & Yes & No & Yes & Yes\\
		\citet{perrin2020fictitious}& Monotone & No & No & No & No & Yes\\
		\citet{geist2021concave}& \thead{Potential\\ \&Monotone}& No & No& No &Yes &Yes\\
		\citet{perolat2021scaling}& Monotone & Yes & Yes & Yes & No& No \\
		\citet{fabian2022learning}& Monotone & Yes & Yes & Yes & No& No \\
		\textbf{Our work} & \textbf{Monotone} & \textbf{Yes}& \textbf{Yes}& \textbf{Yes}& \textbf{Yes}& \textbf{Yes} \\
		\hline 
	\end{tabular}
	\label{table:compare}
	\vspace{-6mm}
\end{table} 

\section{Related Works}
\ac{mfg}s were proposed by \citet{huang2006large} and \citet{lasry2007mean} to model the interactions among a set of  homogeneous agents. In recent years, learning the \ac{ne} of the \ac{mfg}s formulated by discrete-time \ac{mdp}s has attracted a lot of interest. There is a large body of works that design and analyze algorithms for the \ac{mfg}s under  contraction conditions~\citep{anahtarci2019fitted,anahtarci2022q,cui2021approximately,xie2021learning,zaman2022oracle,yardim2022policy}. Typically, these works design reinforcement learning algorithms to approximate the contraction operators in \ac{mfg}s, and the \ac{ne} is learned by iteratively applying this operator. In contrast, another line of works focuses on the \ac{mfg}s under monotone conditions. Motivated by \citet{lasry2007mean}, the transition kernels in these works are independent of the mean fields. \citet{perrin2020fictitious} proposed and analyzed the continuous-time fictitious play algorithm, which dynamically weighs the past mean fields and the best response policies to derive new mean fields and policies. With the additional potential assumption, \citet{geist2021concave} derived the convergence rate for the discrete-time fictitious play algorithm. \citet{perolat2021scaling} then proposed the continuous-time policy mirror descent algorithm but only proved the asymptotic convergence, i.e., the consistency. In summary, there is no convergence rate result for any discrete-time algorithm for \ac{mfg}s under the monotone condition. In addition, the relationship between the contraction conditions and the monotone conditions is not clear from existing works, but they complement each other.

To capture the potential heterogeneity among agents, \ac{gmfg}s have been  proposed by \citet{parise2019graphon} in static settings as an extension of \ac{mfg}s. The heterogeneous interactions among agents are represented by graphons. \citet{aurell2022stochastic,caines2019graphon} then extended the \ac{gmfg}s to the continuous-time setting, where the existence and the uniqueness of \ac{ne} were established. \citet{vasal2020master} formulated  discrete-time \ac{gmfg}s and provided way to compute the \ac{ne} with master equations. With the precise underlying graphons values, \citet{cui2021learning} proposed algorithms to learn the \ac{ne} of \ac{gmfg}s with the contraction condition by modifying \ac{mfg}s learning algorithms. \citet{fabian2022learning} considered the monotone \ac{gmfg} and proposed the continuous-time policy mirror descent algorithm to learn the \ac{ne}. However, only consistency was provided in the latter two works.

\textbf{Notation} Let $[N]:=\{1,\cdots,N\}$. Given a measurable space $(\Omega, \calF)$, we denote the collection of all the measures and the probability measures on $(\Omega,\calF)$ as $\calM(\Omega)$ and $\Delta(\Omega)$, respectively. For a metric space $(\calX,\|\cdot\|)$, we use $C(\calX)$ and $C_{b}(\calX)$ to denote the set of all continuous  functions and the set of all bounded continuous  functions on $\calX$, respectively. For a measurable space $(\calX,\calF)$ and two distributions $P,Q\in\Delta(\calX)$, the total variation distance between them is defined as $\tv(P,Q):=\sup_{A\in\calF}|P(A)-Q(A)|$. A sequence of measures $\{\mu_{n}\}$ on $\calX$ is said to {\em  converge weakly} to a measure $\mu$ if $\int_{\calX}g(x)\mu_{n}(\rmd x)\rightarrow\int_{\calX}g(x)\mu(\rmd x)$ for all $g\in C_{b}(\calX)$. %If not explicitly specified, we take the convergence of measures to be the sense of $\ell_{1}$.
\section{Preliminaries}
\subsection{Graphon Mean-Field Games}
We consider a \ac{gmfg} that is defined through a tuple $(\calI,\calS,\calA,H,P,r,W,\mu_{1})$. The  horizon (or length) of the game  is denoted as $H$. In \ac{gmfg}, we have infinite players, each corresponding to a point $\alpha\in\calI=[0,1]$. The state and action space of them are the same, denoted as $\calS\subseteq\bbR^{d_{\rms}}$ and $\calA\subseteq\bbR^{d_{\rma}}$ respectively. The interaction among players is captured by \emph{graphons}. Graphons are symmetric functions that map $[0,1]^{2}$ to $[0,1]$. Symmetry here refers to that $W(\alpha,\beta)=W(\beta,\alpha)$ for all $\alpha,\beta\in[0,1]$. We denote the set of graphons as $\calW=\{W:[0,1]^{2}\rightarrow[0,1]\,|\, W \text{ is symmetric.}\}$. The set of graphons  of the game is $W=\{W_{h}\}_{h=1}^{H}$ with $W_{h}\in\calW$. The state transition and reward of each player are influenced by the collective behavior of all the other players through an {\em aggregate} $z\in\calM(\calS)$. At time $h\in[H]$, we denote the state distribution of player $\beta\in\calI$ as $\mu_{h}^{\beta}\in\Delta(\calS)$. The aggregate for player $\alpha\in\calI$ with the underlying graphon $W_{h}\in\calW$ is then defined as
\begin{align}
    z^{\alpha}_{h}=\int_{0}^{1} W_{h}(\alpha,\beta)\mu_{h}^{\beta} \, \rmd \beta. \label{eq:1}
\end{align}
The transition kernels $P=\{P_{h}\}_{h=1}^{H}$ of the game are functions $P_{h}:\calS\times\calA\times\calM(\calS)\rightarrow\calS$ for all $h\in[H]$. At time $h$, if player $\alpha$ takes action $a_{h}^{\alpha}\in\calA$ at state $s_{h}^{\alpha}\in\calS$, her state will transition according to $s_{h+1}^{\alpha}\sim P_{h}(\cdot\,|\,s_{h}^{\alpha},a_{h}^{\alpha},z_{h}^{\alpha})$. The reward functions are denoted as  $r_{h}:\calS\times\calA\times\calM(\calS)\rightarrow[0,1]$ for all $h\in[H]$. We note that the players in \ac{gmfg} are heterogeneous. This means that different players will, in general, receive different aggregates from other players. The distribution $\mu_{1}\in\Delta(\calS)$ is the initial state distribution for all the players.
A policy for an player $\alpha$ is $\pi^{\alpha}=\{\pi_{h}^{\alpha}\}_{h=1}^{H}\in\Pi^{H}$, where $\pi_{h}^{\alpha}:\calS\rightarrow\Delta(\calA)$ takes action based only on  the current state, and $\Pi$ is the set of all these policies. A policy for all the players $\pi^{\calI}\in\tilde{\Pi}=\Pi^{H\times\calI}$ is the collection of the policies of each player, i.e, $\pi^{\calI}=\{\pi^{\alpha}\}_{\alpha\in\calI}$. In the following,  we denote the state distributions of all the players at time $h$  and the state distributions of all the players at any time (distribution flow) respectively as $\mu_{h}^{\calI}=\{\mu_{h}^{\alpha}\}_{\alpha\in\calI}$  and $\mu^{\calI}=\{\mu_{h}^{\calI}\}_{h=1}^{H}\in\tilde{\Delta}=\Delta(\calS)^{H\times\calI}$. Eqn.~\eqref{eq:1} shows that the aggregate $z_{h}^{\alpha}$ is a function of $\mu_{h}^{\calI}$ and $W_{h}$, so to make this dependence explicit, we also write it as $z_{h}^{\alpha}(\mu_{h}^{\calI},W_{h})$.

We consider the entropy-regularized \ac{gmfg}. It has been shown that the regularization results in  policy gradient algorithms converging faster~\citep{shani2020adaptive,cen2022fast}. In this game, the rewards of each player are the sum of the original rewards and the negative entropy of the policy multiplied by a constant. In a $\lambda$-regularized \ac{gmfg} ($\lambda\geq 0$), the value function and the action-value function of player $\alpha$ with policy $\pi^{\alpha}$ on the \ac{mdp} induced by the distribution flow $\mu^{\calI}$ are defined as
\begin{align}
    V_{h}^{\lambda,\alpha}(s,\pi^{\alpha},\mu^{\calI})&=\bbE^{\pi^{\alpha}}\bigg[\sum_{t=h}^{H}r_{t}\big(s^{\alpha}_{t},a^{\alpha}_{t},z^{\alpha}_{t}(\mu_{t}^{\calI},W_{t})\big)-\lambda\log\pi_{t}^{\alpha}(a_{t}^{\alpha}\,|\,s_{t}^{\alpha})\,\bigg|\, s^{\alpha}_{h}=s\bigg]\label{eq:vf},\\
    Q_{h}^{\lambda,\alpha}(s,a,\pi^{\alpha},\mu^{\calI})&=r_{h}\big(s,a,z^{\alpha}_{h}(\mu_{h}^{\calI},W_{h})\big)+\bbE^{\pi^{\alpha}}\big[V_{h+1}^{\lambda,\alpha}(s_{h+1}^{\alpha},\pi^{\alpha},\mu^{\calI})\,|\, s_{h}^{\alpha}=s,a_{h}^{\alpha}=a\big]\label{eq:avf}
\end{align}
for all $h\in[H]$, where the expectation $\bbE^{\pi^{\alpha}}$ is taken with respect to the stochastic process induced by implementing policy $\pi^{\alpha}$ on the \ac{mdp} induced by $\mu^{\calI}$. The cumulative reward function of player $\alpha$ is defined as $J^{\lambda,\alpha}(\pi^{\alpha},\mu^{\calI})=\bbE_{\mu_{1}}[V_{1}^{\lambda,\alpha}(s,\pi^{\alpha},\mu^{\calI})]$. Then the notion of an \ac{ne} is defined as follows.
\begin{definition}\label{def:gmfgne}
    An \ac{ne} of the $\lambda$-regularized \ac{gmfg} is a pair $(\pi^{*,\calI},\mu^{*,\calI})\in \tilde{\Pi}\times\tilde{\Delta}$ that satisfies: (i) (player rationality) $J^{\lambda,\alpha}(\pi^{*,\alpha},\mu^{*,\calI})=\max_{\tilde{\pi}^{\alpha}\in\Pi^{H}} J^{\lambda,\alpha}(\tilde{\pi}^{\alpha},\mu^{*,\calI})$ for all $\alpha\in\calI$ up to a zero-measure set on $\calI$ with respect to the Lebesgue measure. (ii) (Distribution consistency) The distribution flow $\mu^{*,\calI}$ is equal to the distribution flow induced by implementing the policy $\pi^{*,\calI}$.
\end{definition}
Similar to the \ac{ne} for the finite-player games, the \ac{ne} of the $\lambda$-regularized \ac{gmfg} requires that the policy of each player is optimal. However,  in \ac{gmfg}s,  the optimality is with respect to $\mu^{*,\calI}$.

\subsection{Mean-Field Games}
As an important subclass of \ac{gmfg}, \ac{mfg} corresponds to the \ac{gmfg} with {\em constant} graphons, i.e, $W(\alpha,\beta)=p$ for all $\alpha,\beta\in\calI$. \ac{mfg}s involve infinite \emph{homogeneous} players. All the players employ the same policy and thus share the same distribution flow. The aggregate in Eqn.~\ref{eq:1} degenerates to $z^{\alpha}_{h}=\int_{0}^{1} p\cdot\mu_{h}^{\beta} \rmd \beta=p\cdot\mu_{h}$ for all $\alpha\in\calI$. Here $\mu_{h}$ is the state distribution of a representative player. Thus, an \ac{mfg} is denoted by a tuple $(\bar{\calS},\bar{\calA},\bar{H}, \bar{P},\bar{r},\mu_{1})$. The state space, the action space, and the horizon are respectively denoted as $\bar{\calS}$, $\bar{\calA}$, and $\barH$. Here, the transition kernels $\barP=\{\barP_{h}\}_{h=1}^{H}$ are functions $\barP_{h}:\calS\times\calA\times\Delta(\calS)\rightarrow\calS$, and reward functions $\barr_{h}:\calS\times\calA\times\Delta(\calS)\rightarrow[0,1]$ for all $h\in[H]$. In the \ac{mfg}, all the players adopt the {\em same} policy $\pi=\{\pi_{h}\}_{h=1}^{H}$ where $\pi_{h}\in\Pi$. The value function and the action-value function in the $\lambda$-regularized \ac{mfg} with the underlying distribution flow $\mu=\{\mu_{h}\}_{h=1}^{H}\in\Delta(\calS)^{H}$ can be similarly defined as Eqn.~\eqref{eq:vf} and \eqref{eq:avf} respectively as follows
\begin{align*}
    \barV_{h}^{\lambda}(s,\pi,\mu)&=\bbE^{\pi}\bigg[\sum_{t=h}^{H}\barr_{t}(s_{t},a_{t},\mu_{t})-\lambda\log\pi_{t}(a_{t}\,|\,s_{t})\,\bigg|\, s_{h}=s\bigg],\\
    \barQ_{h}^{\lambda}(s,a,\pi,\mu)&=\barr_{h}(s,a,\mu_{h})+\bbE^{\pi}\big[\barV_{h+1}^{\lambda}(s_{h+1},\pi,\mu)\,|\, s_{h}=s,a_{h}=a\big]
\end{align*}
for all $h\in[H]$. The cumulative reward is $\barJ^{\lambda}(\pi,\mu)=\bbE_{\mu_{1}}[\barV_{1}^{\lambda}(s,\pi,\mu)]$. The notion of \ac{ne} can be similarly derived as follows.
\begin{definition}\label{def:mfgne}
    An \ac{ne} of the $\lambda$-regularized \ac{mfg} is a pair $(\pi^{*},\mu^{*})\in \Pi^{H}\times\Delta(\calS)^{H}$ that satisfies: (i) (player rationality) $J^{\lambda}(\pi^{*},\mu^{*})=\max_{\tilde{\pi}\in\Pi^{H}} J^{\lambda}(\tilde{\pi},\mu^{*})$. (ii) (Distribution consistency) The distribution flow $\mu^{*}$ is equal to the distribution flow induced by the policy $\pi^{*}$.
\end{definition}
\begin{remark}
    Compared with Definition~\ref{def:gmfgne}, the definition of \ac{ne} in \ac{mfg} only involves the policy and the distribution flow of a single representative player, since the agents are homogeneous in \ac{mfg}s.
\end{remark}

\section{Existence of the \ac{ne}s in Regularized \ac{gmfg}s and \ac{mfg}s}
We now state some assumptions to demonstrate the existence of a  \ac{ne} for  $\lambda$-regularized \ac{gmfg}s.
\begin{assumption}\label{assump:compact}
    The state space $\calS$ is compact, and the action space $\calA$ is finite.
\end{assumption}
This assumption imposes rather  mild constraints on $\calS$ and $\calA$. In  real-world applications, the states are usually physical quantities and thus reside in a compact set. For the action space, many deep reinforcement learning algorithms discretize the potential continuous action sets into finite sets~\citep{lowe2017multi,mordatch2018emergence}.
\begin{assumption}\label{assump:glip}
    The graphons $W_{h}$ for $h\in[H]$ are continuous functions.
\end{assumption}
The stronger version of this assumption (Lipschitz continuity) is widely adopted in \ac{gmfg} works~\citep{cui2021learning,fabian2022learning}. It helps us to build the continuity of the transition kernels and rewards with respect to players. %This assumption can be relaxed to piece-wise Lipschitzness, and our results in this section can easily generalize.
\begin{assumption}\label{assump:conti}
    For all $h\in[H]$, the reward function $r_{h}(s,a,z)$ is continuous on $\calS\times\calA\times\calM(\calS)$, that is if $(s_{n},a_{n},z_{n})\rightarrow(s,a,z)$ as $n\rightarrow\infty$, then $r_{h}(s_{n},a_{n},z_{n})\rightarrow r_{h}(s,a,z)$. The transition kernel $P_{h}(\cdot\,|\, s,a,z)$ is weakly continuous in $\calS\times\calA\times\calM(\calS)$, that is if $(s_{n},a_{n},z_{n})\rightarrow(s,a,z)$ as $n\rightarrow\infty$, $P_{h}(\cdot\,|\,s_{n},a_{n},z_{n})\rightarrow P_{h}(\cdot\,|\,s,a,z)$ weakly.
\end{assumption}
This assumption states the continuity of the models, i.e., the transition kernels and the rewards, as functions of the state, action, and aggregate. We note that the Lipschitz continuity assumption of the model in the previous works implies that  our assumption is satisfied~\citep{cui2021learning,fabian2022learning}. Next, we state the existence result of regularized \ac{gmfg}.
\begin{theorem}\label{thm:exist}
    Under Assumptions~\ref{assump:compact}, \ref{assump:glip} and \ref{assump:conti}, for all $\lambda\geq 0$, the $\lambda$-regularized \ac{gmfg} $(\calI,\calS,\calA,H,P,r,W,\mu_{1})$ admits an \ac{ne} $(\pi^{\calI},\mu^{\calI})\in \tilde{\Pi}\times\tilde{\Delta}$.
\end{theorem}
This theorem strengthens  previous existence results in \cite{cui2021learning} and \cite{fabian2022learning} in two aspects. First, our assumptions are weaker. These two existing works require a finite state space and the Lipschitz continuity of the models. In contrast, we only need a compact state space and the model continuity. Second, their results only hold for the unregularized case ($\lambda=0$), whereas ours holds for any $\lambda\geq0$. In the proof of Theorem~\ref{thm:exist}, we construct a \ac{mfg} from the given \ac{gmfg} and show that an \ac{ne} of the constructed \ac{mfg} can be converted to an \ac{ne} of the \ac{gmfg}. Then we  prove the existence of \ac{ne} in the constructed regularized \ac{mfg}. Our existence result for the regularized \ac{mfg} is also a significant addition to the \ac{mfg} literature. %Before reaching the existence result for \ac{mfg}, we make an important remark concerning the proof  of Theorem~\ref{thm:exist}.
\begin{remark}
    Although we show that an \ac{ne} of the constructed \ac{mfg} can be converted to an \ac{ne} of \ac{gmfg}, this proof does not imply that \ac{gmfg} forms a subclass of or is equivalent to \ac{mfg}. This is because we have only demonstrated  the relationship between the \ac{ne}s of these two games, but the exact realizations of the \ac{gmfg} and the conceptually constructed \ac{mfg} may differ. It means that the sample paths of these two games may not be the same, which include the realizations of the states, actions, and rewards of all the players.
\end{remark}
We next state the assumption needed for the \ac{mfg}.
\begin{assumption}\label{assump:mfg}
    The \ac{mfg} $(\bar{\calS},\bar{\calA},\bar{H}, \bar{P},\bar{r},\mu_{1})$ satisfies that: (i) The state space $\bar\calS$ is compact, and the action space $\bar\calA$ is finite. (ii) The reward functions $\barr_{h}(s,a,\mu)$ for $h\in[H]$ are continuous on $\calS\times\calA\times\Delta(\calS)$. The transition kernels are weakly continuous on $\calS\times\calA\times\Delta(\calS)$; that is if $(s_{n},a_{n},\mu_{n})\rightarrow (s,a,\mu)$ as $n\rightarrow\infty$, $\barP_{h}(\cdot\,|\,s_{n},a_{n},\mu_{n})\rightarrow \barP_{h}(\cdot\,|\,s,a,\mu)$ weakly.
\end{assumption}
Then the existence of the \ac{ne} is stated as follows.
\begin{theorem}\label{thm:mfgexist}
    Under Assumption~\ref{assump:mfg}, the $\lambda$-regularzied \ac{mfg} $(\bar{\calS},\bar{\calA},\bar{H}, \bar{P},\bar{r},\mu_{1})$ admits an \ac{ne} $(\pi,\mu)\in\Pi^{H}\times\Delta(\calS)^{H}$ for any $\lambda\geq 0$.
\end{theorem}
Our result in Theorem~\ref{thm:mfgexist} imposes weaker conditions than previous works~\citep{cui2021approximately,anahtarci2022q} to guarantee the existence of an \ac{ne}. These existing works prove the existence of \ac{ne} by assuming a \emph{contractive} property and the finiteness of the state space. They also require a strong Lispchitz assumption~\citep{anahtarci2022q}, where the Lipschitz constants of the models should be small enough. In contrast, we only require the continuity assumption in Theorem~\ref{thm:mfgexist}. This is due to our analysis of the operator for the regularized \ac{mfg} and the use of Kakutani fixed point theorem~\cite{guide2006infinite}.

\section{Learning \ac{ne} of Monotone \ac{gmfg}s}
In this section, we focus on \ac{gmfg}s with transition kernels that are independent of the aggregate $z$, i.e., $P_{h}:\calS\times\calA\rightarrow\calS$ for $h\in[H]$. This model is motivated by the seminal work \cite{lasry2007mean}, where the state evolution in  continuous-time \ac{mfg} is characterized  by the Fokker--Plank equation. However, the form of the Fokker--Plank equation results in  the state transition of each player  being independent of other players. This model is also widely accepted in the discrete-time \ac{gmfg} and \ac{mfg} literature~\citep{fabian2022learning,perrin2020fictitious,perolat2021scaling}.
\subsection{Monotone \ac{gmfg}}\label{sec:mono}
We first generalize the notion of \emph{monotonicity} from multi-population \ac{mfg}s in \cite{perolat2021scaling} to \ac{gmfg}s.
\begin{definition}[Weakly Monotone Condition]
    A \ac{gmfg} is said to be \emph{weakly monotone} if for any $\rho^{\calI},\tilde{\rho}^{\calI}\in \Delta(\calS\times\calA)^{\calI}$ and their marginalizations on the states $\mu^{\calI},\tilde{\mu}^{\calI}\in \Delta(\calS)^{\calI}$, we have
    \begin{align*}
        \int_{\calI}\sum_{a\in\calA}\int_{\calS} \big(\rho^{\alpha}(s,a)-\tilde{\rho}^{\alpha}(s,a)\big)\Big(r_{h}\big(s,a,z_{h}^{\alpha}(\mu^{\calI},W_{h})\big)-r_{h}\big(s,a,z_{h}^{\alpha}(\tilde{\mu}^{\calI},W_{h})\big)\Big)\, \rmd s\, \rmd \alpha \leq 0 
    \end{align*}
    for all $h\in[H]$, where $W_{h}$ is the underlying graphon. It is \emph{strictly weakly monotone} if the inequality is strict when $\rho^{\calI}\neq \tilde{\rho}^{\calI}$.
\end{definition}
In two \ac{mdp}s induced by the distribution flows $\mu^{\calI}$ of $\pi^{\calI}$ and $\tilde{\mu}^{\calI}$ of $\tilde{\pi}^{\calI}$, the weakly monotone condition states that we can achieve higher rewards at stage $h\in[H]$ by swapping the policies. This condition has two important implications. The first is the uniqueness of the \ac{ne}.
\begin{proposition}\label{prop:uniquene}
    Under Assumptions~\ref{assump:compact}, \ref{assump:glip}, and \ref{assump:conti}, a strictly weakly monotone $\lambda$-regularized \ac{gmfg} has a unique \ac{ne} for any $\lambda\geq 0$ up to a zero-measure set on $\calI$ with respect to the Lebesgue measure.
\end{proposition}
In the following, we denote this unique \ac{ne} as $(\pi^{*,\calI},\mu^{*,\calI})$, and we aim to learn this \ac{ne}. The second implication concerns the relationship between the cumulative rewards of two policies.
\begin{proposition}\label{prop:mono}
    If a $\lambda$-regularized \ac{gmfg} satisfies the weakly monotone condition, then for any two policies $\pi^{\calI}$, $\tilde{\pi}^{\calI}\in\tilde{\Pi}$ and their induced distribution flows $\mu^{\calI}, \tilde{\mu}^{\calI}\in\tilde{\Delta}$, we have
    \begin{align*}
        \int_{0}^{1}J^{\lambda,\alpha}(\pi^{\alpha},\mu^{\calI})+J^{\lambda,\alpha}(\tilde{\pi}^{\alpha},\tilde{\mu}^{\calI})-J^{\lambda,\alpha}(\tilde{\pi}^{\alpha},\mu^{\calI})-J^{\lambda,\alpha}(\pi^{\alpha},\tilde{\mu}^{\calI})\, \rmd\alpha\leq 0
    \end{align*}
    If the  $\lambda$-regularized \ac{gmfg} satisfies the strictly weakly monotone condition, then the inequality is strict when $\pi^{\calI}\neq \tilde{\pi}^{\calI}$.
\end{proposition}
Proposition~\ref{prop:mono} shows that if we have two policies, we can improve the cumulative rewards on the \ac{mdp} induced by these policies by swapping the policies. This implies an important property of the \ac{ne} $(\pi^{*,\calI},\mu^{*,\calI})$. Since $\pi^{*,\calI}$ is optimal on the \ac{mdp} induced by $\mu^{*,\calI}$, we have $\int_{0}^{1}J^{\lambda,\alpha}(\pi^{*,\alpha},\mu^{*,\calI})\, \rmd\alpha\geq \int_{0}^{1}J^{\lambda,\alpha}(\pi^{\alpha},\mu^{*,\calI})\, \rmd\alpha$ for any $\pi^{\calI}\in\tilde{\Pi}$. Then Proposition~\ref{prop:mono} shows that $\int_{0}^{1}J^{\lambda,\alpha}(\pi^{*,\alpha},\mu^{\calI})\, \rmd\alpha\geq \int_{0}^{1}J^{\lambda,\alpha}(\pi^{\alpha},\mu^{\calI})\, \rmd\alpha$ for any policy $\pi^{\calI}$ and the distribution flow $\mu^{\calI}$ it induces. This means that the \ac{ne} policy gains cumulative rewards not less than any policy $\pi^{\calI}$ on the \ac{mdp} induced by $\pi^{\calI}$. This motivates the design of our \ac{ne} learning algorithm.

\subsection{Policy Mirror Descent Algorithm for Monotone \ac{gmfg}}\label{sec:algodesign}
In this section, we introduce the algorithm to learn the \ac{ne}, which is called \ac{mgmfgpmd} and whose pseudo-code is outlined in  Algorithm~\ref{algo:pmd}. It consists of  three main steps. The first step estimates the action-value function (Line~3). We need to evaluate the action-value function of a policy on the \ac{mdp} induced by itself. This estimate can be obtained for each player independently by playing the $\pi_{t}^{\calI}$ several times. We assume  access to a sub-module for this and quantify the estimation error in our analysis. The second step is the policy mirror descent (Line~4). Given $\lambda\eta_{t}<1$, This step can be equivalently formulated as
\begin{align*}
    \hpi_{t+1,h}^{\alpha}(\cdot\,|\,s)=\argmax_{p\in\Delta(\calA)}\frac{\eta_{t}}{1-\lambda\eta_{t}}\Big[\big\langle \hatQ_{h}^{\lambda,\alpha}(s,\cdot,\pi_{t}^{\alpha},\mu_{t}^{\calI}),p \big\rangle-\lambda R(p)\Big]-\kl\big(p\|\pi_{t,h}^{\alpha}(\cdot\,|\,s)\big)\quad \forall\, s\in\calS,
\end{align*}
where $R(p) = \langle p , \log p\rangle$ is the negative entropy function. This step aims to improve the performance of the policy $\pi_{t}^{\calI}$ on its own induced \ac{mdp}. Intuitively, since the policy $\pi^{*,\calI}$ in \ac{ne} performs better than $\pi_{t}^{\calI}$ on the \ac{mdp} induced by $\mu_{t}^{\calI}$ as shown in Section~\ref{sec:mono}, the improved policy $\pi_{t+1}^{\calI}$ should be closer to $\pi^{*,\calI}$ than $\pi_{t}^{\calI}$. The third step mixes the current policy with the uniform policy (Line~5) to encourage exploration.

\ac{mgmfgpmd}  is different from previous \ac{ne} learning algorithms for monotone \ac{gmfg} in \cite{perolat2021scaling,fabian2022learning} in three different ways. First, \ac{mgmfgpmd}  is designed to learn the \ac{ne} of the $\lambda$-regularized \ac{gmfg} with $\lambda>0$, whereas other algorithms learn the \ac{ne} of the {\em unregularized} \ac{gmfg}s. As a result, our policy improvement (Line~4) discounts the previous policy as $(\pi_{t,h}^{\alpha})^{1-\lambda\eta_{t}}$, but other algorithms retain $\pi_{t,h}^{\alpha}$. Second, our algorithm is discrete-time and thus amenable  for  practical implementation. However, other provably efficient algorithms evolve in continuous time. Finally, \ac{mgmfgpmd}  encourages  exploration in Line 5, which is important for the theoretical analysis. Such a step is missing in other algorithms.

\begin{algorithm}[t]
	\caption{\ac{mgmfgpmd}}
	\textbf{Procedure:}
	\begin{algorithmic}[1]\label{algo:pmd}
        \STATE Initialize $\pi_{1,h}^{\alpha}(\cdot\,|\,s)=\unif(\calA)$ for all $s\in\calS$,$h\in[H]$ and $\alpha\in\calI$.
        \FOR{$t=1,2,\cdots,T$}
            \STATE Compute the action-value function $\hatQ_{h}^{\lambda,\alpha}(s,a,\pi_{t}^{\alpha},\mu_{t}^{\calI})$ for all $\alpha\in\calI$ and $h\in[H]$, where $\mu_{t}^{\calI}$ is the distribution flow induced by $\pi_{t}^{\calI}$.
            \STATE $\hpi_{t+1,h}^{\alpha}(\cdot\,|\,s)\propto \big(\pi_{t,h}^{\alpha}(\cdot\,|\,s)\big)^{1-\lambda\eta_{t}}\exp\big(\eta_{t} \hatQ_{h}^{\lambda,\alpha}(s,a,\pi_{t}^{\alpha},\mu_{t}^{\calI})\big)$ for all $\alpha\in\calI$ and $h\in[H]$
            \STATE $\pi_{t+1,h}^{\alpha}(\cdot\,|\,s)=(1-\beta_{t})\hpi_{t+1,h}^{\alpha}(\cdot\,|\,s)+\beta_{t}\unif(\calA)$
        \ENDFOR
        \STATE Output $\bar{\pi}^{\calI}=\unif\big(\pi_{[1:T]}^{\calI}\big)$
	\end{algorithmic}
\end{algorithm}
\vspace{-2mm}
\subsection{Theoretical Analysis for \ac{mgmfgpmd} with Estimation Oracle}
This section provides  theoretical analysis for the \ac{mgmfgpmd} algorithm given an action-value function oracle in Line~3, i.e., $\hatQ_{h}^{\lambda,\alpha}=Q_{h}^{\lambda,\alpha}$. We denote the unique \ac{ne} of the $\lambda$-regularized \ac{gmfg} as $(\pi^{*,\calI},\mu^{*,\calI})$. For any policy $\pi^{\calI}$, we measure the distance to the policy $\pi^{*,\calI}$ of \ac{ne} using
\begin{align*}D(\pi^{\calI})=\int_{0}^{1}\sum_{h=1}^{H}\bbE_{\mu_{h}^{*,\alpha}}\Big[\kl\big(\pi_{h}^{*,\alpha}(\cdot\,|\,s_{h}^{\alpha})\|\pi_{h}^{\alpha}(\cdot\,|\,s_{h}^{\alpha})\big)\Big]\,\rmd\alpha.
\end{align*}
%\vspace{-1mm}
This metric measures the weighted KL divergence between policy $\pi^{\calI}$ and the \ac{ne} policy, and the weights are the \ac{ne} distribution flow $\mu^{*,\calI}$.
%Then our result for the \ac{mgmfgpmd} is as follows.
\begin{theorem}\label{thm:monoopt}
    Assume that the \ac{gmfg} is strictly weakly monotone and we have an action-value function oracle. Let  
     $\eta_{t}=\eta=O(T^{-1/2})$ and $\beta_{t}=\beta=O(T^{-1})$ in \ac{mgmfgpmd}. For any $\lambda>0$ we have
     \vspace{-2mm}
    \begin{align*}
        D\Bigg(\frac{1}{T}\sum_{t=1}^{T}\pi_{t}^{\calI}\Bigg)=O\bigg(\frac{\lambda\log^{2} T}{T^{1/2}}\bigg).
    \end{align*}
\end{theorem}
%\vspace{-4mm}
Theorem~\ref{thm:monoopt} provides the first convergence rate result for a discrete-time algorithm on strictly weakly monotone \ac{gmfg}s under mild assumptions. In contrast, \citet{perolat2021scaling,fabian2022learning} only consider the continuous-time policy mirror descent, which is difficult for the practical implementation, and only provide  the asymptotic consistency results. \citet{geist2021concave} derive exploitability results for a fictitious play algorithm but require the potential structure and the Lipschitzness of the reward function. Our proof for Theorem~\ref{thm:monoopt} mainly exploits  properties of \ac{ne} discussed in Section~\ref{sec:mono}. Concretely, we use the fact that  the policy mirror descent procedure reduces the distance between the policy iterate and the \ac{ne} as $D(\pi_{t+1}^{\calI})-D(\pi_{t}^{\calI})\approx \int_{0}^{1} J^{\lambda,\alpha}(\pi_{t}^{\alpha},\mu_{t}^{\calI})-J^{\lambda,\alpha}(\pi^{*,\alpha},\mu_{t}^{\calI})\, \rmd\alpha$. Thus, the policy iterate becomes  closer to the \ac{ne}. However, the discretization error and the exploration influence (Line 5) also appear, requiring additional care to show that $D(\pi_{t}^{\calI})$, in general, decreases. %The details are in the appendix.

\subsection{Theoretical Analysis for \ac{mgmfgpmd} with General Function Classes}
In this section, we remove the requirement that one is given oracle access to an action-value function and we estimate it in Line 3 of \ac{mgmfgpmd}  using general function classes. We consider the action-value function class $\calF=\calF_{1}\times\cdots\times\calF_{H}$, where $\calF_{h}\subseteq\{f:\calS\times\calA\rightarrow[0,(H-h+1)(1+\lambda\log|\calA|)]\}$ is the class of action-value functions at time $h\in[H]$. Then we estimate the action-value functions using Algorithm~\ref{algo:est}. 

\begin{algorithm}[t]
	\caption{Estimation of Action-value Function}
	\textbf{Procedure:}
	\begin{algorithmic}[1]\label{algo:est}
        \STATE Sample $N$ players $\{i/N\}_{i=1}^{N}\subseteq[0,1]$.
        \STATE The $i^{\rm th}$ player implements $\pi_{t}^{\rmb,i}$ for $i\in[N]$, and the other players implement $\pi_{t}^{\calI}$.
        \STATE Collect data $\{(s_{\tau,h}^{i},a_{\tau,h}^{i},r_{\tau,h}^{i})\}_{i,\tau,h=1}^{N,K,H}$ of sampled players from $K$ independent episodes.
        \STATE Initialize $\hatV_{H+1}^{\lambda,i}(s,a)=0$ for all $s\in\calS$, $a\in\calA$ and $i\in[N]$.
        \FOR{ time $h=H,\cdots,1$}
            \FOR{ Player $i=1,\cdots,N$(in parallel) }
                \STATE $\hatQ_{h}^{\lambda,i}=\argmin_{f\in\calF_{h}}\sum_{\tau=1}^{K}\big(f(s_{\tau,h}^{i},a_{\tau,h}^{i})-r_{\tau,h}^{i}-\hatV_{h+1}^{\lambda,i}(s_{\tau,h+1}^{i})\big)^{2}$.
                \STATE $\hatV_{h}^{\lambda,i}(s)=\langle \hatQ_{h}^{\lambda,i}(s,\cdot),\pi_{t,h}^{i/N}(\cdot,s) \rangle-\lambda R\big(\pi_{t,h}^{i/N}(\cdot,s)\big)$.
            \ENDFOR
        \ENDFOR
        \STATE Output $\{\hatQ_{h}^{\lambda,i}\}_{i,h=1}^{N,H}$.
	\end{algorithmic}
\end{algorithm}
This algorithm mainly involves two steps. The first is involves data collection (Line~3). Here we assign policies to players and collect data from their interactions.  We let the sampled $N$ players implement behavior policies $\{\pi_{t}^{\rmb,i}\}_{i=1}^{N}$, which can be different from $\{\pi_{t}^{i/N}\}_{i=1}^{N}$. This will not change the aggregate $z_{h}^{\alpha}(\mu_{h}^{\calI},W_{h})$ for any $\alpha\in\calI$, since only a zero-measure set of players change their policies. The second is the action-value function estimation (Lines~7 and~8). The action-value function is selected based on the previous value function, and the value function is updated from the derived estimate. This can be implemented in parallel for all  players. We highlight that the estimation analysis cannot leverage results from  general non-parametric regression~\citep{wainwright2019high}, since the response variable $\hatV_{h+1}^{\lambda,i}$ is \emph{not} independent of $s_{\tau,h+1}^{i}$ in our setting. %To derive the theoretical guarantee for our estimates, we first state an assumption.
\begin{assumption}[Realizability]\label{assump:real}
    For any policy $\pi^{\calI}\in\tilde{\Pi}$ and the induced distribution flow $\mu^{\calI}\in\tilde{\Delta}$, we have $Q_{h}^{\lambda,\alpha}(\cdot,\cdot,\pi^{\alpha},\mu^{\calI})\in\calF_{h}$ for $h\in[H]$.
\end{assumption}
This assumption ensures that we can find the nominal action-value function in the function class. For a policy $\pi\in\Pi^{H}$ and a function $f:\calS\times\calA\rightarrow\bbR$, we define the operator $(\calT_{h}^{\pi}f)(s,a)=\bbE_{s^{\prime}\sim P_{h}(\cdot|s,a)}[\langle f(s^{\prime},\cdot),\pi_{h+1}(\cdot|s^{\prime})\rangle-\lambda R(\pi_{h+1}(\cdot|s^{\prime}))]$. For a policy $\pi^{\calI}$ and the induced distribution flow $\mu^{\calI}$, we have $Q_{h}^{\lambda,\alpha}(s,a,\pi^{\alpha},\mu^{\calI})=r_{h}(s,a,z_{h}^{\alpha})+(\calT_{h}^{\pi^{\alpha}}Q_{h+1}^{\lambda,\alpha})(s,a)$.
\begin{assumption}[Completeness]\label{assump:complete}
    For any policy $\pi^{\calI}\in\tilde{\Pi}$ and the induced distribution flow $\mu^{\calI}\in\tilde{\Delta}$, we have that for all $f\in\calF_{h+1}$, $r_{h}(\cdot,\cdot,z_{h}^{\alpha}(\mu_{h}^{\calI},W_{h}))+(\calT_{h}^{\pi^{\alpha}}f)\in\calF_{h}$ for all $\alpha\in\calI$, $h\in[H-1]$.
\end{assumption}
This completeness assumption ensures that the estimates from $\calF$ also satisfy the relationship between nominal action-value functions through $\calT_{h}^{\pi^{\alpha}}$. These realizability and completeness assumptions are widely adopted in the off-policy evaluation and offline reinforcement learning literature~\citep{uehara2022review,xie2021bellman}. 
\begin{assumption}\label{assump:rlip}
    The reward functions $\{r_{h}\}_{h=1}^{H}$ are Lipschitz in $z$, i.e., $|r_{h}(s,a,z)-r_{h}(s,a,z^{\prime})|\leq L_{r}\|z-z^{\prime}\|_{1}$ for all $s\in\calS,a\in\calA,h\in[H]$. The graphons $W_{h}$ for $h\in[H]$ are Lipschitz continuous functions, i.e., there exists a constant $L=L_{W}>0$ (depending only on $W=\{W_{h}\}_{h=1}^{H}$) such that $|W_{h}(\alpha,\beta)-W_{h}(\alpha^{\prime},\beta^{\prime})|\leq L_{W}(|\alpha-\alpha^{\prime}|+|\beta-\beta^{\prime}|)$ for all $\alpha,\alpha^{\prime},\beta,\beta^{\prime}\in[0,1]$, and $h\in[H]$.
\end{assumption}
The Lipschitzness assumption is common in the \ac{gmfg} works~\citep{parise2019graphon,carmona2022stochastic,cui2021learning}. It helps us to approximate the action-value function of a player by that of sampled players. We denote the state distributions of player $i$ induced by policy $\pi_{t}^{\rmb,i}$ as $\mu_{t}^{\rmb,i}$. Then we require the behavior policies $\{\pi_{t}^{\rmb,i}\}_{i=1}^{N}$ to satisfy the following requirements.
\begin{assumption}\label{assump:concen}
    For any $t\in[T]$, the behavior policies   explore sufficiently. More precisely,  for any policy $\pi\in\Pi^{H}$ and   induced distributions $\mu\in\Delta(\calS)^{H}$, we have $\sup_{s\in\calS,a\in\calA}\pi_{h}(a\,|\,s)/\pi_{t,h}^{\rmb,i}(a\,|\,s)\leq C_{1}$ and $\sup_{s\in\calS}\rmd\mu_{h}(s)/\rmd\mu_{t,h}^{\rmb,i}(s)\leq C_{2}$ for $h\in[H]$ and $i\in[N]$ where  $C_{1},C_{2}>0$ are   constants.
\end{assumption}
This assumption guarantees that the behavior policies explore the actions that may be adopted by $\pi_{t}^{\calI}$ and $\pi^{*,\calI}$. Such an assumption is widely adopted in offline reinforcement learning and off-policy evaluation works~\citep{uehara2022review,xie2021bellman}. 
\begin{theorem}\label{thm:optest}
    Assume that  the \ac{gmfg} is weakly monotone and that Assumptions~\ref{assump:real}, \ref{assump:complete}, \ref{assump:rlip}, and \ref{assump:concen} hold. Let  $\eta_{t}=\eta=O(T^{-1/2})$ and $\beta_{t}=\beta=O(T^{-1})$ in Algorithm~\ref{algo:pmd} (\ac{mgmfgpmd}). Then with probability at least $1-\delta$, Algorithms~\ref{algo:pmd} and \ref{algo:est} yield
    \begin{align*}
        D\Bigg(\frac{1}{T}\sum_{t=1}^{T}\pi_{t}^{\calI}\Bigg)=O\bigg(\frac{\lambda\log^{2} T}{\sqrt{T}}+C_{1}C_{2}\frac{H^{3/2}B_{H}^{2}}{\lambda\sqrt{K}}\log\frac{TNH\cdot\calN_{\infty}(5B_{H}/K,\calF_{[H]})}{\delta}+\frac{ H\log T}{N}\bigg),
    \end{align*}
     where $B_{H}=H(1+\lambda\log|\calA|)$, and $\calN_{\infty}(5B_{H}/K,\calF_{[H]})=\max_{h\in[H]}\calN_{\infty}(5B_{H}/K,\calF_{h})$ is the $\ell_{\infty}$-covering number of the function class.
\end{theorem}
\vspace{-2mm}
The error in Theorem~\ref{thm:optest} consists of both the optimization and estimation errors. The optimization error corresponds to the first term, which also appears in Theorem~\ref{thm:monoopt}. The estimation error consists of the generalization error and the approximation error, in the second and third terms respectively.  When the function class $\calF$ is finite, this term scales as $O(K^{-1/2})$, which originates from the fact that we estimate the action-value function from the empirical error instead of its population counterpart. The approximation error scales as $O(N^{-1})$. This term originates from the fact that the action-value function of player $\alpha$ is approximated by that of the sampled player near $\alpha$. To learn a policy that is at most $\varepsilon>0$ far from the \ac{ne}, we can set $T=\tilO(\varepsilon^{-2})$, $K=\tilO(\varepsilon^{-2})$, and $N=\tilO(\varepsilon^{-1})$, which in total results in $TK=\tilO(\varepsilon^{-4})$ episodes of online plays.
\vspace{-2mm}
\section{Experiments}
\vspace{-2mm}
In this section, we conduct experiments to corroborate our theoretical findings. We run different algorithms on the Beach Bar problem~\citet{perrin2020fictitious,fabian2022learning}. The underlying graphons are set to \ac{sbm} and exp-graphons. The details of experiments are deferred to Appendix~\ref{app:experiment}. Since the \ac{ne}s of the games are not available, we adopt the \emph{exploitability} to measure the proximity between a policy and the \ac{ne}. For a policy $\pi^{\calI}$ and its induced distribution flow $\mu^{\calI}$, the exploitability for the $\lambda$-regularized \ac{gmfg} is defined as
\begin{align*}
    {\rm{Exploit}}(\pi^{\calI})=\int_{0}^{1}\max_{\tilde{\pi}\in\Pi^{H}}J^{\lambda,\alpha}(\tilde{\pi},\mu^{\calI})-J^{\lambda,\alpha}(\pi^{\alpha},\mu^{\calI})\rmd\alpha.
\end{align*}

\begin{figure}[t]
	\centering
	\subfigure[Beach Bar problem with \ac{sbm} graphons]{
	\begin{minipage}[t]{0.44\linewidth}
	\centering
	\includegraphics[width=2.3in]{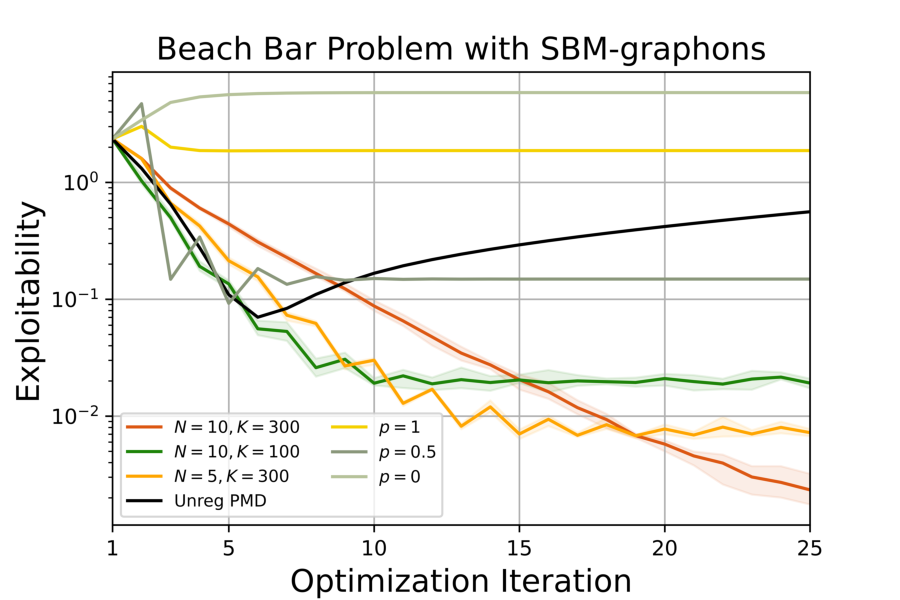}
	%\caption{fig2}
	\end{minipage}%
	\label{fig:simu_sbm}
	}%
	\hspace{1cm}
	\centering
	\subfigure[Beach Bar problem with exp-graphons.]{
	\begin{minipage}[t]{0.44\linewidth}
	\centering
	\includegraphics[width=2.3in]{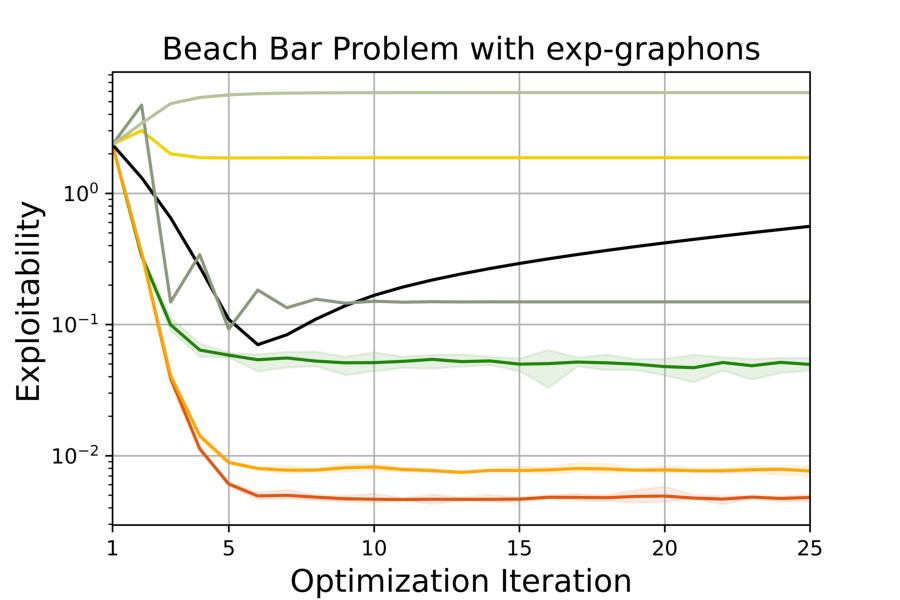}
	%\caption{fig2}
	\end{minipage}%
	\label{fig:simu_exp}
	}%
	\vspace{-0.3cm}
	\centering
	\caption{Simulation results for Beach Bar problem with \ac{sbm} and exp-graphons.}
	\label{fig:simu}
	\vspace{-7mm}
\end{figure}

First, the experimental results demonstrate the necessity of modelling the heterogeneity of agents. Figure~\ref{fig:simu} demonstrates the performance degradation of approximating \ac{gmfg} by \ac{mfg}. Here we let the agents play in the \ac{gmfg} with constant graphons $W_{h}(\alpha,\beta)=p$ for $p\in \{0,0.5,1\}$. The  agents have oracle access to the action-value function. We observe that  this  approximation results in gross errors for learning the \ac{ne}s of \ac{gmfg}s with non-constant graphons.

Second, the experiments show that the algorithms designed for unregularized \ac{gmfg} cannot learn the \ac{ne} of regularized \ac{gmfg}. We implement the discrete-time version of the algorithm in \citet{fabian2022learning}; results marked  ``Unreg PMD'' show that the exploitability first decreases and then increases. In line with the discussion in Section~\ref{sec:algodesign}, this originates from keeping too much gradient knowledge in previous iterates $\pi_{t}^{\calI}$. The gradient of the policy is largely correct in the several initial iterations, but a large amount of past knowledge results in it deviating in later iterations, since the past knowledge  accumulates. In contrast, our  algorithm {\em discounts} the past knowledge as $(\pi_{t}^{\calI})^{1-\lambda\eta_{t}}$.

Finally, the results indicate the influence of action-value function estimation. In the experiments, we run our algorithm when $N=5,K=300$, $N=10,K=100$, and $N=10,K=300$. Figure~\ref{fig:simu} shows that the algorithm with $N=10,K=300$ can achieve a smaller error than the algorithms both with $N=10,K=100$ and $N=5,K=300$. This is in agreement with  Theorem~\ref{thm:optest}.

\section{Conclusion}\label{app:conclusion}
In this paper, we focused on two fundamental problems of $\lambda$-regularized \ac{gmfg}. Firstly, we established the existence of \ac{ne}. This result greatly weakened the conditions in the previous works. Secondly, the provably efficient \ac{ne} learning algorithms were proposed and analyzed in the weakly monotone \ac{gmfg} motivated by \cite{lasry2007mean}. The convergence rate of \ac{mgmfgpmd} features the first performance guarantee of discrete-time algorithm without extra conditions in monotone \ac{gmfg}s. We leave the lower bound of this problem to the future works.

\textbf{Acknowledgements}
Fengzhuo Zhang and Vincent Tan acknowledge funding by the Singapore Data Science Consortium (SDSC) Dissertation Research Fellowship,  the Singapore Ministry of Education Academic Research Fund (AcRF) Tier 2 under grant number A-8000423-00-00, and AcRF Tier 1 under grant numbers A-8000980-00-00 and A-8000189-01-00. Zhaoran Wang acknowledges National Science Foundation (Awards 2048075, 2008827, 2015568, 1934931), Simons Institute (Theory of Reinforcement Learning), Amazon, J.P. Morgan, and Two Sigma for their support.

\bibliographystyle{plainnat}%abbrvnat%unsrt
\bibliography{ref}

\newpage

\clearpage
\appendix

\begin{center}
{\Large {\bf Appendix for \\ ``Learning Regularized Monotone Graphon Mean-Field Games''}}
\end{center}

\setcounter{equation}{0}
\counterwithin*{equation}{section}

\renewcommand{\theequation}{\thesection.\arabic{equation}}

\section{Detailed Explanations of Table~\ref{table:compare}}\label{app:compare}
We first explain Table~\ref{table:compare} column by column. The first column lists the conditions required by each work. Although the detailed statements of these conditions are usually different, these conditions can be largely categorized into contraction conditions and monotone conditions. Here `potential' means the extra potential reward structure required in \cite{geist2021concave}. 

`No population manipulation' means that during the learning process, the distribution flow is indeed induced by the current policy. For example, \cite{xie2021learning} and \cite{perrin2020fictitious} mix the current distribution flow with the previous ones to form the distribution flow required by the next step. In contrast, the distribution flows required in algorithms in \cite{perolat2021scaling}, \cite{fabian2022learning} and our work are those induced by the policies in each step. 

`Online playing' means that the algorithms can be implemented with the data collected from the online playing of agents. In general, the algorithms that do not require population manipulations can be implemented by letting agents play their policies in the online game. Thus, these algorithms admit online playing. In contrast, \cite{anahtarci2022q}, \cite{perrin2020fictitious} and \cite{geist2021concave} need to solve the optimal policy on specific distribution flows. Thus, they need the access to a simulator for this purpose. 

`Heterogeneity' means the modeling of the heterogeneity among agents. The works for \ac{mfg} only consider homogeneous agents, and thus cannot model the heterogeneity. 

`Discrete-time algorithm' means the provably efficient discrete-time algorithms here. Although some discrete-time algorithms are provided in \cite{perrin2020fictitious},\cite{perolat2021scaling} and \cite{fabian2022learning}, neither the consistency nor the convergence rate is provided therein. 

`Convergence rate' in the final column refers to the convergence rate of both discrete-time algorithms and continuous-time algorithms. \cite{perrin2020fictitious} provides the convergence rate for their continuous-time algorithm, and other works with `Yes' all provide the convergence rate for the discrete-time algorithms. 

In summary, our work provides the first provably efficient discrete-time algorithm in the monotone \ac{gmfg} without any extra conditions. This result deepen our understanding of the monotone \ac{gmfg}s, as a complementary setting of contractive \ac{gmfg}s.

\section{Experiment Details}\label{app:experiment}
We adopt the Beach Bar problem as our \ac{gmfg}. This problem is initially proposed in \citet{perolat2021scaling,perrin2020fictitious} for \ac{mfg} and modified by \citet{fabian2022learning} to \ac{gmfg}. In the Beach Bar problem, Agents can move their towels between locations and try to be close to the bar but also avoid crowded
areas and neighbors in an underlying network. The state space $\calS$ is $\{1,2,\cdots,|\calS|\}$, and we set $|\calS|=10$ in our experiments. The bar is located at $B=|\calS|/2$. The action space is $\calA=\{-1,0,1\}$, which indicates the movement of the towel. The transition kernel is $s_{t+1}^{\alpha}= s_{t}^{\alpha}+a_{t}^{\alpha}+\varepsilon_{t}^{\alpha}$, where $\varepsilon_{t}^{\alpha}$ is the noise that takes $+1$ or $-1$ with probability $1/2$. The reward function is defined as
\begin{align*}
    r_{t}(s_{t}^{\alpha},a_{t}^{\alpha},z_{t}^{\alpha})=\frac{2}{|\calS|}|B-s_{t}^{\alpha}|+\frac{2}{|\calS|}|a_{t}^{\alpha}|-8z_{t}^{\alpha}\text{ for all }t\in[H].
\end{align*}
In our experiments, we regualrize this reward function with $\lambda=1$. The underlying graphons in our experiments are \ac{sbm} and exp-graphons. The exp-graphon is defined as
\begin{align*}
    W_{\theta}^{\rm exp}(\alpha,\beta)=\frac{2\exp(\theta\cdot \alpha\beta)}{1+\exp(\theta\cdot \alpha\beta)}-1,
\end{align*}
where $\theta>0$ is the parameter. In our simulation, we set $\theta=3$. The \ac{sbm} in our experiments has two communities with $70\%$ and $30\%$ population respectively. The inter-community rate is $0.3$, and the intra-community rate is $0.9$. In our experiments, we adopt the exploitability to measure the closeness between a policy and the \ac{ne}. 

Since the Beach Bar problem only involves the finite state and action spaces, our algorithms take the function class $\calF_{h}=\{f:\calS\times\calA\rightarrow [0,H(1+\lambda\log|\calA|)]\}$ for all $h\in[H]$. 

Figure~\ref{fig:simu} is generated from five Monte-Carlo implementations for each algorithm. The error bar in the figure indicates the maximal and the minimal error in the Monte-Carlo. To simulate cases with constant graphons, we implement mirror descent algorithm with the nominal action-value functions, which are directly calculated from the ground-truth transition kernels and reward functions. Thus, there is no error bar for thm. To simulate the policy mirror descent for unregularized \ac{gmfg}, we directly use the code of \cite{fabian2022learning}, and the action-value functions are also calculated from the ground-truth model. Thus, there is no error bar for it, either. Our simulations run on a single Intel(R) Xeon(R) CPU E5-2697 v4 @ 2.30 GHz, and the experiments take about two days,

\section{Proof of Theorem~\ref{thm:exist}}
\begin{proof}[Proof of Theorem~\ref{thm:exist}]
    We prove the existence of \ac{ne} by three steps:
    \begin{itemize}
        \item We construct a $\lambda$-regularized \ac{mfg} based on the $\lambda$-regularized \ac{gmfg}.
        \item We show that we can construct an \ac{ne} of the $\lambda$-regularized \ac{gmfg} from an \ac{ne} of the constructed $\lambda$-regularized \ac{mfg}.
        \item We show that the constructed $\lambda$-regularized \ac{mfg} has \ac{ne} under Assumptions~\ref{assump:lip} and \ref{assump:conti}.
    \end{itemize}
    
    \textbf{Step 1: Construction of a $\lambda$-regularized \ac{mfg}}
    
    The state and action spaces of the $\lambda$-regularized \ac{mfg} is $\bar{\calS}=\calS\times\calI$ and $\calA$ respectively, where $\calI=[0,1]$. Here, we treat the positions of players as a state in \ac{mfg}. The state of the player is denoted as $\bar{s}_{h}=(s_{h},\alpha_{h})\in\bar{\calS}$, and we denote the distribution of the state at time $h$ as $\calL_{h}=\calL(\bar{s}_{h})=\calL(s_{h},\alpha_{h})$, which is the law of state $\bar{s}_{t}$. 
    
    At time $h$, the transition kernel of such \ac{mfg} is $\bar{P}_{h}:\bar{\calS}\times\calA\times\Delta(\bar{\calS})\rightarrow\Delta(\bar{\calS})$. To specify $\bar{P}_{h}$, we first define a function of $\alpha_{h}$ and $\calL_{h}$ as $f(\alpha_{h},\calL_{h},W_{h}):\calI\times \Delta(\calS\times\calI)\times\calW\rightarrow \calM(\calS)$, whose output is a measure supported on $\calS$, i.e.,
    \begin{align}
        \big[f(\alpha_{h},\calL_{h},W_{h})\big](\cdot)=\int_{0}^{1} W_{h}(\alpha_{h},\beta)\calL_{h}(\cdot,\beta)\rmd \beta.\label{eq:5}
    \end{align}
    
    Eqn.~\eqref{eq:5} enables us to define the transition kernel $\bar{P}_{h}$ as
    \begin{align}
        \bar{P}_{h}(\bar{s}_{h+1}\,|\, \bar{s}_{h},a_{h},\calL_{h})&=\delta_{\alpha_{h+1}=\alpha_{h}}\cdot P_{h}\big(s_{h+1}\,|\, s_{h},a_{h},f(\alpha_{h},\calL_{h},W_{h})\big)\nonumber\\
        &=\delta_{\alpha_{h+1}=\alpha_{h}}\cdot P_{h}\Big(s_{h+1}\,|\, s_{h},a_{h},\int_{0}^{1} W_{h}(\alpha_{h},\beta)\calL_{h}(\cdot,\beta)\rmd \beta\Big),\label{eq:6}
    \end{align}
    where $\delta_{\alpha_{h+1}=\alpha_{h}}$ is the Dirac's delta function at $\alpha_{h}$, and $P_{h}$ is the transition kernel of the \ac{gmfg}. The reward function of the \ac{mfg} can be similarly defined as
    \begin{align}
        \bar{r}_{h}(\bar{s}_{h},a_{h},\calL_{h})=r_{h}\Big(s_{h},a_{h},\int_{0}^{1} W_{h}(\alpha_{h},\beta)\calL_{h}(\cdot,\beta)\rmd \beta\Big).\label{eq:7}
    \end{align}
    The initial state distribution of the \ac{mfg} is specified as $\calL_{1}=\mu_{1}\times \unif([0,1])$. The value functions of the $\lambda$-regularized \ac{mfg} are defined as
    \begin{align*}
        \bar{V}^{\lambda}_{h}\big((s,\alpha),\pi,\calL\big)=\bbE^{\pi}\bigg[\sum_{t=h}^{H}\bar{r}_{h}(\bar{s}_{t},a_{t},\calL_{t})-\lambda\log\pi_{t}(a_{t}\,|\,s_{t},\alpha_{t})\,\bigg|\, s_{h}=s,\alpha_{h}=\alpha\bigg],
    \end{align*}
    where the expectation $\bbE^{\pi}$ is taken with respect to the \ac{mdp} $a_{t}\sim\pi_{t}(\cdot\,|\,s_{t},\alpha_{t})$ and $\bar{s}_{t+1}\sim\bar{P}_{t}(\cdot\,|\, \bar{s}_{t},a_{t},\calL_{t})$ for $t\in[H]$. Then the cumulative reward function is defined as
    \begin{align}
        \bar{J}^{\lambda}(\pi,\calL)=\bbE_{\calL_{1}}\big[\bar{V}^{\lambda}_{1}(s,\alpha,\pi,\calL)\big].\label{eq:mfgreward}
    \end{align}
    
    \textbf{Step 2: Construction of the \ac{ne} of the $\lambda$-regularized \ac{gmfg} from the \ac{ne} of the $\lambda$-regularized \ac{mfg}}
    
    In this step, we assume that the $\lambda$-regularized \ac{mfg} defined in Eqn.~\eqref{eq:6} and \eqref{eq:7} admits an \ac{ne} $(\tilde{\pi},\tilde{\calL})$, which is defined in Definition~\ref{def:mfgne}, replacing the discounted reward function therein by the reward defined in Eqn.~\eqref{eq:mfgreward}. We will construct a policy and distribution flow pair $(\pi^{\calI},\mu^{\calI})$ of the $\lambda$-regularized \ac{gmfg} from $(\tilde{\pi},\tilde{\calL})$ and show that $(\pi^{\calI},\mu^{\calI})$ is indeed an \ac{ne} of the $\lambda$-regularized \ac{gmfg}.
    
    We construct the policy and distribution flow pair as $\pi_{h}^{\alpha}(\cdot\,|\, s)=\tilde{\pi}_{h}(\cdot\,|\,s,\alpha)$ and $\mu_{h}^{\alpha}(\cdot)=\tilde{\calL}_{h}(\cdot,\alpha)$ for all $h\in[H]$, $s\in\calS$, and $\alpha\in\calI$. To prove that $(\pi^{\calI},\mu^{\calI})$ is an \ac{ne} of \ac{gmfg}, we need to show: (i) $\mu^{\calI}$ is induced by the policy $\pi^{\calI}$, i.e., $\Gamma_{2}(\pi^{\calI},W)=\mu^{\calI}$. (ii) $\pi^{\calI}$ is the optimal policy given $\mu^{\calI}$ for all the players.
    
    We use induction to prove (i). Define $\tilde{\mu}^{\calI}=\Gamma_{2}(\pi^{\calI},W)$. We will show that $\tilde{\mu}^{\calI}=\mu^{\calI}$. For $h=1$, we have $\mu_{1}^{\prime,\alpha}=\mu_{1}=\calL_{1}(\cdot,\alpha)=\mu_{1}^{\alpha}$ for all $\alpha\in\calI$. Assume that $\mu_{h}^{\prime,\alpha}=\mu_{h}^{\alpha}$ for all $\alpha\in\calI$, for time $h+1$ and any $s^{\prime}\in\calS$ and $\alpha\in\calI$, we have
    \begin{align*}
        \mu_{h+1}^{\prime,\alpha}(s^{\prime})&=\int_{\calS}\int_{\calA}\mu_{h}^{\prime,\alpha}(s)\pi_{h}^{\alpha}(a\,|\, s)P_{h}\Big(s^{\prime}\, |\, s,a,\int_{0}^{1}W_{h}(\alpha,\beta)\mu_{h}^{\prime,\alpha}\rmd \beta\Big)\rmd a \rmd s\\
        &=\int_{\calS}\int_{\calA}\int_{0}^{1}\tilde{\calL}_{h}(s,\alpha^{\prime})\tilde{\pi}_{h}(a\,|\, s,\alpha^{\prime})P_{h}\Big(s^{\prime}\, |\, s,a,\int_{0}^{1}W_{h}(\alpha^{\prime},\beta)\tilde{\calL}_{h}(\cdot,\beta)\rmd \beta\Big)\delta_{\alpha=\alpha^{\prime}}\rmd\alpha^{\prime}\rmd a \rmd s\\
        &=\tilde{\calL}_{h+1}(s^{\prime},\alpha),
    \end{align*}
    where the first equation follows from the definition of $\tilde{\mu}^{\calI}$, the second equation follows from the fact that $\mu_{h}^{\prime,\alpha}=\mu_{h}^{\alpha}$ and $\mu_{h}^{\alpha}(\cdot)=\tilde{\calL}_{h}(\cdot,\alpha)$, and the last equation follows from the definition of $\tilde{\calL}$. Then (i) results from the fact that $\mu_{h+1}^{\alpha}(\cdot)=\tilde{\calL}_{h+1}(\cdot,\alpha)= \mu_{h+1}^{\prime,\alpha}(\cdot)$.
    
    To prove (ii), we compare the \ac{mdp}s given $\tilde{\calL}$ and $\mu^{\calI}$ in \ac{mfg} and \ac{gmfg}. The \ac{mdp} for the player $\alpha$ in \ac{gmfg} is specified by the transition kernel $s_{h+1}^{\alpha}\sim P_{h}(\cdot\,|\, s_{h}^{\alpha},a_{h}^{\alpha},\int_{0}^{1}W_{h}(\alpha,\beta)\mu_{h}^{\beta}\rmd\beta)$, and the reward function $r_{h}(s_{h}^{\alpha},a_{h}^{\alpha},\int_{0}^{1}W_{h}(\alpha,\beta)\mu_{h}^{\beta}\rmd\beta)$. We want to prove that $V_{1}^{\lambda,\alpha}(s,\pi^{\alpha},\mu^{\calI})\geq V_{1}^{\lambda,\alpha}(s,\bar{\pi}^{\alpha},\mu^{\calI})$ for all $s\in\calS$, and $\bar{\pi}^{\alpha}\in\Pi$.
    
    Since $\tilde{\pi}$ is optimal with respect to $\tilde{\calL}$, we have $\bar{V}_{1}^{\lambda}(s,\alpha,\tilde{\pi},\tilde{L})\geq \bar{V}_{1}^{\lambda}(s,\alpha,\bar{\pi},\tilde{L})$ for all $s\in\calS$, $\alpha\in\calI$, and policy $\bar{\pi}$. Given $s_{1}=s$ and $\alpha_{1}=\alpha$, the \ac{mdp} in \ac{mfg} is specified by the transition kernel $s_{h+1}\sim P_{h}(s_{h+1}\,|\, s_{h},a_{h},\int_{0}^{1}W_{h}(\alpha,\beta)\tilde{\calL}_{t}(\cdot,\beta)\rmd\beta)$, $\alpha_{h+1}=\alpha$, and the reward function $r_{h}(s_{h}^{\alpha},a_{h}^{\alpha},\int_{0}^{1}W_{h}(\alpha,\beta)\tilde{\calL}_{t}(\cdot,\beta)\rmd\beta)$. We note that these two \ac{mdp}s are the same, and $V_{1}^{\lambda,\alpha}(s,\pi^{\alpha},\mu^{\calI})=\bar{V}_{1}^{\lambda}(s,\alpha,\tilde{\pi},\tilde{L})$. This proves the claim $(ii)$.
    
    \textbf{Step 3: Prove the existence of \ac{ne} in the constructed \ac{mfg} under Assumptions~\ref{assump:glip} and \ref{assump:conti}.}
    
    In order to prove the existence of \ac{ne} in the constructed \ac{mfg}, we only need to verify Assumption~\ref{assump:mfg} in Theorem~\ref{thm:mfgexist}.
    
    We first verify Assumption~\ref{assump:mfg} (1) and (2) hold. Our reward functions $\barr_{h}$ are bounded, $\calA$ is finite, and the state space $\calS\times\calI$ is compact.
    
    For Assumption~\ref{assump:mfg} (3) and (4), we only need to verify that the reward function in Eqn.~\eqref{eq:7} is continuous and the transition kernel in Eqn.~\eqref{eq:6} is continuous with respect to total variation.  Since $r_{h}$ is continuous, we only need to prove that $\int_{0}^{1} W_{h}(\alpha_{h},\beta)\calL_{h}(\cdot,\beta)\rmd \beta$ is continuous for the continuity of $\bar{r}_{h}$. In the following, we make use of the fact that the convergence in total variation implies the weak convergence. 
    
    Given two sequences $\{\alpha_{n}\}$ and $\{\calL_{n}\}$ such that $\alpha_{n}\rightarrow\alpha$ and $\calL_{n}$ converges to $\calL$ in total variation, we have
    \begin{align}
        &\int_{\calS}\bigg|\int_{0}^{1}W_{h}(\alpha_{n},\beta)\calL_{n}(s,\beta)\rmd\beta-\int_{0}^{1}W_{h}(\alpha,\beta)\calL(s,\beta)\rmd\beta\bigg|\rmd s\nonumber\\
        &\quad\leq\int_{\calS}\bigg|\int_{0}^{1}W_{h}(\alpha_{n},\beta)\calL_{n}(s,\beta)\rmd\beta-\int_{0}^{1}W_{h}(\alpha,\beta)\calL_{n}(s,\beta)\rmd\beta\bigg|\rmd s\nonumber\\
        &\quad\qquad +\int_{\calS}\bigg|\int_{0}^{1}W_{h}(\alpha,\beta)\calL_{n}(s,\beta)\rmd\beta-\int_{0}^{1}W_{h}(\alpha,\beta)\calL(s,\beta)\rmd\beta\bigg|\rmd s\nonumber\\
        &\quad\leq \int_{\calS}\bigg|\int_{0}^{1}W_{h}(\alpha_{n},\beta)\calL_{n}(s,\beta)\rmd\beta-\int_{0}^{1}W_{h}(\alpha,\beta)\calL_{n}(s,\beta)\rmd\beta\bigg|\rmd s\nonumber\\
        &\quad\qquad+\int_{\calS}\int_{0}^{1}\big|\calL_{n}(s,\beta)-\calL(s,\beta)\big|\rmd\beta\rmd s,\label{ieq:1}
    \end{align}
    where the first inequality follows from the triangle inequality. Since $\alpha_{n}\rightarrow\alpha$ and the uniform continuity of graphons, the first term in the right-hand side of inequality \eqref{ieq:1} tends to $0$. Since $\calL_{n}$ converges to $\calL$ in total variation, the second term in the right-hand side of inequality \eqref{ieq:1} tends to $0$. Thus, the reward function is continuous, which verifies Assumption~\ref{assump:mfg} (3). For any $g\in C_{b}(\calS\times\calI)$, given four sequences $\{s_{n}\}$, $\{\alpha_{n}\}$, $\{a_{n}\}$, and $\{L_{n}\}$ such that $s_{n}\rightarrow s$, $\alpha_{n}\rightarrow \alpha$, $a_{n}\rightarrow a$, and $L_{n}$ weakly converges to $L$, we have
    \begin{align}
        &\bigg|\int_{0}^{1}\int_{\calS}g(s^{\prime},\alpha^{\prime})\big[\bar{P}_{h}(s^{\prime},\alpha^{\prime}\,|\, \bar{s}_{h},a_{h},\calL_{h})-\bar{P}_{h}(s^{\prime},\alpha^{\prime}\,|\, \bar{s},a,\calL)\big]\rmd s^{\prime}\rmd\alpha^{\prime}\bigg|\nonumber\\
        &\quad =\bigg|\int_{\calS}g(s^{\prime},\alpha_{n})P_{h}\Big(s^{\prime}\,\Big|\, s_{n},a_{n},z_{h}^{\alpha_{n}}(\calL_{n},W_{h})\Big)\rmd s^{\prime}-\int_{\calS}g(s^{\prime},\alpha)P_{h}\Big(s^{\prime}\,\Big|\, s,a,z_{h}^{\alpha}(\calL,W_{h})\Big)\rmd s^{\prime}\bigg|\nonumber\\
        &\quad\leq \bigg|\int_{\calS}g(s^{\prime},\alpha_{n})P_{h}\Big(s^{\prime}\,\Big|\, s_{n},a_{n},z_{h}^{\alpha_{n}}(\calL_{n},W_{h})\Big)\rmd s^{\prime}-\int_{\calS}g(s^{\prime},\alpha)P_{h}\Big(s^{\prime}\,\Big|\, s,a,z_{h}^{\alpha_{n}}(\calL_{n},W_{h})\Big)\rmd s^{\prime}\bigg|\nonumber\\
        &\quad\qquad +\bigg|\int_{\calS}g(s^{\prime},\alpha)P_{h}\Big(s^{\prime}\,\Big|\, s,a,z_{h}^{\alpha_{n}}(\calL_{n},W_{h})\Big)\rmd s^{\prime}-\int_{\calS}g(s^{\prime},\alpha)P_{h}\Big(s^{\prime}\,\Big|\, s,a,z_{h}^{\alpha}(\calL,W_{h})\Big)\rmd s^{\prime}\bigg|,\label{ieq:2}
    \end{align}
    where $z_{h}^{\alpha_{n}}(\calL_{n},W_{h})=\int_{0}^{1} W_{h}(\alpha_{n},\beta)\calL_{n}(\cdot,\beta)\rmd \beta$, $z_{h}^{\alpha}(\calL,W_{h})=\int_{0}^{1} W_{h}(\alpha,\beta)\calL(\cdot,\beta)\rmd \beta$, the equation follows from Eqn.~\eqref{eq:6}, and the inequality follows from the triangle inequality. Since $g\in C_{b}(\calS\times\calI)$ is a continuous function on a compact set and $\alpha_{n}\rightarrow\alpha$, the first term in the right-hand side of inequality~\eqref{ieq:2} tends to $0$. Since inequality \eqref{ieq:1} proves that $z_{h}^{\alpha_{n}}(\calL_{n},W_{h})$ converges to $z_{h}^{\alpha}(\calL,W_{h})$ in $\ell_{1}$ and $g(\cdot,\alpha)\in C_{b}(\calS)$, the second term in the right-hand side of inequality~\eqref{ieq:2} tends to $0$. Thus, the transition kernel is continuous, which verifies Assumption~\ref{assump:mfg}. It concludes the verification of Assumption~\ref{assump:mfg} in Theorem~\ref{thm:mfgexist}. Thus, we conclude the proof of Theorem~\ref{thm:exist}.
\end{proof}

\section{Proof of Theorem~\ref{thm:mfgexist}}
\begin{proof}[Proof of Theorem~\ref{thm:mfgexist}]
    We prove the existence of \ac{ne} by two steps:
    \begin{itemize}
        \item We construct an operator $\Gamma$ that is defined for the state-action distribution flow and show that we can construct the \ac{ne} from the fixed point of this operator.
        \item We show that the fixed point set of the operator $\Gamma$ is not empty.
    \end{itemize}
    \textbf{Step 1: Construction of an operator $\Gamma$.}
    
    Without the loss of generality, we assume that $r_{h}:\calS\times\calA\times\Delta(\calS)\rightarrow[0,1]$ for $h\in[H]$. Define constants $L_{h}=(H-h+1)(1+\lambda\log|\calA|)$ for $h\in[H]$. Given the continuous functions set $C(\calS)$, we define the $L_{h}-$bounded continuous function set $C_{h}=\{f\in C(\calS)\,|\, \sup_{s\in\calS}|f(s)|\leq L_{h}\}$ and the product of them $\calC=\prod_{h=1}^{H}C_{h}(\calS)$. Given a constant $0<\sigma<1$, we equip $\calC$ with the metric $\rho(u,v)=\sum_{h=1}^{H}\sigma^{-h}\|u_{h}-v_{h}\|_{\infty}$ for any $u,v\in\calC$ and $\|f\|_{\infty}=\sup_{s\in\calS}|f(s)|$. Then $(\calC,\rho)$ is complete.
    
    We define the state-action distribution flow set as $\Xi=\prod_{h=1}^{H}\Delta(\calS\times\calA)$. For ease of notation, we denote the marginalization of any $\nu\in\Xi$ on $\calS$ as $\bar{\nu}_{h}(s)=\sum_{a\in\calA}\nu_{h}(s,a)$.
    
    For any $\nu\in\Xi$, define an operator $T_{h}^{\nu}$ acting on $\calS\rightarrow\bbR$ as
    \begin{align*}
        T_{h}^{\nu}u(s)&=\sup_{p\in\Delta(\calA)} \sum_{a\in\calA}p(a)\barr_{h}(s,a,\bar{\nu}_{h})-\lambda R(p)+\sum_{a\in\calA}\int_{\calS}p(a)\barP_{h}(s^{\prime}\,|\,s,a,\bar{\nu}_{h})u(s^{\prime})\rmd s^{\prime} \text{ for }h\in[H-1],\\ %\label{eq:10}
        T_{H}^{\nu}u(s)&=\sup_{p\in\Delta(\calA)} \sum_{a\in\calA}p(a)\barr_{H}(s,a,\bar{\nu}_{H})-\lambda R(p),%\label{eq:11}
    \end{align*}
    where $R(\cdot)$ is the negative entropy. When $\lambda=0$, the supremum is taken with respect to the action $a\in\calA$, and the following proposition can be similarly built for $\lambda=0$.
    
    \begin{proposition}\label{prop:mfgcontract}
        Let $\nu\in\Xi$ be arbitrary distribution flow. For all $h\in[H]$, $T_{h}^{\nu}$ maps $C_{h+1}(\calS)$ into $C_{h}(\calS)$. In addition, for any $u,v\in C_{h+1}$, we have $\|T_{h}^{\nu}u-T_{h}^{\nu}v\|_{\infty}\leq \|u-v\|_{\infty}$.
    \end{proposition}
    \begin{proof}[Proof of Proposition~\ref{prop:mfgcontract}]
        See Appendix~\ref{app:mfgcontract}.
    \end{proof}
    We then define an operator $T^{\nu}:\calC\rightarrow\calC$ as $(T^{\nu}u)_{h}=T^{\nu}_{h}u_{h+1}$ for all $h\in[H]$. We then have
    \begin{align}
        \rho(T^{\nu}u,T^{\nu}v)=\sum_{h=1}^{H}\sigma^{-h}\|T^{\nu}_{h}u_{h+1}-T^{\nu}_{h}v_{h+1}\|_{\infty}\leq \sum_{h=1}^{H}\sigma^{-h}\|u_{h+1}-v_{h+1}\|_{\infty}\leq \sigma\rho(u,v),\label{ieq:5}
    \end{align}
    where the first inequality results from Proposition~\ref{prop:mfgcontract}. Thus, $T^{\nu}$ is a contraction on $(\calC,\rho)$, and it has an unique fixed point. For any state-action distribution $\nu\in\Xi$, we use $\barV_{h}^{\lambda,\nu}:\calS\rightarrow\bbR$ for $h\in[H]$ to denote the value functions of the optimal policy in the $\lambda$-regularized \ac{mdp} induced by $\nu$ as $\barr_{h}(s,a,\bar{\nu}_{h})$ and $\barP_{h}(\cdot\,|\,s,a,\bar{\nu}_{h})$ for $h\in[H]$. Then the theory of Markov process shows that \cite[Theorem 14.1, Theorem 17.1]{hindererfoundations}
    \begin{proposition}\label{prop:optcondition}
        For any $\nu\in\Xi$, $\barV^{\lambda,\nu}=(\barV_{h}^{\lambda,\nu})_{h=1}^{H}$ is the unique fixed point of $T^{\nu}$. A policy $\pi\in\Pi^{H}$ is optimal if and only if the following equation holds for any $h\in[H]$ and $\mu_{h}^{+}-a.s.$ state $s\in\calS$, where $\mu^{+}$ is the distribution of states when implementing $\pi$ on the \ac{mdp} induced by $\bar{\nu}$.
        \begin{align*}
            \sum_{a\in\calA}p(a)\barr_{h}(s,a,\bar{\nu}_{h})-\lambda R(p)+\sum_{a\in\calA}\int_{\calS}p(a)\barP_{h}(s^{\prime}\,|\,s,a,\bar{\nu}_{h})\barV_{h+1}^{\lambda,\nu}(s^{\prime})\rmd s^{\prime}=T_{h}^{\nu}\barV_{h+1}^{\lambda,\nu}(s).
        \end{align*}
    \end{proposition}
    
    For any $\nu\in\Xi$, we define the sets
    \begin{align*}
        A(\nu)&=\bigg\{\xi\in\Xi\,\bigg|\,\bar{\xi}_{1}=\mu_{1},\, \bar{\xi}_{h+1}(\cdot)=\sum_{a\in\calA}\int_{\calS}\barP_{h}(\cdot\,|\,s,a,\bar{\nu}_{h})\nu_{h}(\rmd s, a)\text{ for all }h\in[H-1]\bigg\},\\
        B(\nu)&=\bigg\{\xi\in\Xi\,\bigg|\,\text{For all }h\in[H], \, \sum_{a\in\calA}\xi_{h}(a\,|\,s)\barr_{h}(s,a,\bar{\nu}_{h})-\lambda R\big(\xi_{h}(\cdot\,|\,s)\big)\\
        &\quad\qquad+\sum_{a\in\calA}\int_{\calS}\xi_{h}(a\,|\,s)\barP_{h}(s^{\prime}\,|\,s,a,\bar{\nu}_{h})\barV_{h+1}^{\lambda,\nu}(s^{\prime})\rmd s^{\prime}=\barV_{h}^{\lambda,\nu}(s),\,\bar{\xi}_{h}\text{-a.s.}\bigg\},\\
        \Gamma(\nu)&=A(\nu)\cap B(\nu).
    \end{align*}
    We note that $\Gamma(\nu)\neq\emptyset$, since set $A(\nu)$ imposes constraints on $\bar{\xi}_{h}$ while $B(\nu)$ imposes constraints on $\xi_{h}(\cdot\,|\,s)$. We say that $\nu$ is a fixed point of $\Gamma$ if $\nu\in\Gamma(\nu)$.
    
    \textbf{Step 2: The existence of the fixed point of $\Gamma$.}
    
    \begin{proposition}\label{prop:gammapolicy}
        Suppose that $\Gamma$ has a fixed point $\nu\in\xi$. Then we construct a policy as: $\pi_{h}(\cdot\,|\,s)=\nu_{h}(\cdot\,|\,s)$ for all $s\in\supp(\bar{\nu}_{h})$ and $h\in[H]$, $\pi_{h}(\cdot\,|\,s)$ for $s\notin\supp(\bar{\nu}_{h})$ can be arbitrarily defined. Then the pair $(\pi,\bar{\nu})$ is an \ac{ne} of the $\lambda$-regularized \ac{mfg}.
    \end{proposition}
    \begin{proof}[Proof of Proposition~\ref{prop:gammapolicy}]
        Since $\nu\in A(\nu)$, the distribution consistency condition in Definition~\ref{def:mfgne} holds. Since $\nu\in B(\nu)$, the policy defined in Proposition~\ref{prop:gammapolicy} satisfies the optimality condition in Proposition~\ref{prop:optcondition}, which verifies the player rationality condition in Definition~\ref{def:mfgne}.
    \end{proof}
    
    \begin{proposition}\label{prop:closedgraph}
        The graph of $\Gamma$, i.e., $\text{Gr}(\Gamma)=\{(\nu,\xi)\in\Xi\times\Xi\,|\, \xi\in\Gamma(\nu)\}$ is closed.
    \end{proposition}
    \begin{proof}[Proof of Proposition~\ref{prop:closedgraph}]
        See Appendix~\ref{app:closedgraph}.
    \end{proof}
    
    The existence of the fixed point of operator $\Gamma$ follows from the Kakutani's Theorem~\cite[Corollary 17.55 ]{guide2006infinite}. We note that the existence of the \ac{ne} is the direct result of Proposition~\ref{prop:optcondition}. This concludes the proof of Theorem~\ref{thm:mfgexist}.
    
\end{proof}
\section{Proof of Theorem~\ref{thm:monoopt}}\label{app:optproof}
\begin{proof}[Proof of Theorem~\ref{thm:monoopt}]
According to the policy update procedures Line~$5$ and Line~$6$ in Algorithm~\ref{algo:pmd}, we have that
\begin{align}
    &\kl\big(\pi_{h}^{*,\alpha}(\cdot\,|\,s_{h}^{\alpha})\|\pi_{t+1,h}^{\alpha}(\cdot\,|\,s_{h}^{\alpha})\big)\nonumber\\
    &\quad\leq \kl\big(\pi_{h}^{*,\alpha}(\cdot\,|\,s_{h}^{\alpha})\|\hpi_{t+1,h}^{\alpha}(\cdot\,|\,s_{h}^{\alpha})\big)+\beta/(1-\beta)\nonumber\\
    &\quad\leq -\eta\big\langle \hatQ_{h}^{\lambda,\alpha}(s_{h}^{\alpha},\cdot,\pi_{t}^{\alpha},\mu_{t}^{\calI})-\lambda\log\pi_{h}^{\alpha}(\cdot\,|\,s_{h}^{\alpha}),\pi_{h}^{*,\alpha}(\cdot\,|\,s_{h}^{\alpha})-\pi_{t,h}^{\alpha}(\cdot\,|\,s_{h}^{\alpha})\big\rangle\nonumber\\
    &\quad\qquad+\kl\big(\pi_{h}^{*,\alpha}(\cdot\,|\,s_{h}^{\alpha})\|\pi_{t,h}^{\alpha}(\cdot\,|\,s_{h}^{\alpha})\big)+\frac{1}{2}\eta^{2}\bigg(H+\lambda H\log|\calA|+\lambda\log\frac{|\calA|}{\beta}\bigg)^{2}+\frac{\beta}{1-\beta}\nonumber\\
    &\quad\leq -\eta\big\langle Q_{h}^{\lambda,\alpha}(s_{h}^{\alpha},\cdot,\pi_{t}^{\alpha},\mu_{t}^{\calI})-\lambda\log\pi_{h}^{\alpha}(\cdot\,|\,s_{h}^{\alpha}),\pi_{h}^{*,\alpha}(\cdot\,|\,s_{h}^{\alpha})-\pi_{t,h}^{\alpha}(\cdot\,|\,s_{h}^{\alpha})\big\rangle\nonumber\\
    &\quad\qquad+\kl\big(\pi_{h}^{*,\alpha}(\cdot\,|\,s_{h}^{\alpha})\|\pi_{t,h}^{\alpha}(\cdot\,|\,s_{h}^{\alpha})\big)+\frac{1}{2}\eta^{2}\bigg(H+\lambda H\log|\calA|+\lambda\log\frac{|\calA|}{\beta}\bigg)^{2}+\frac{\beta}{1-\beta}+\eta\varepsilon_{h}^{\alpha},\label{ieq:3}
\end{align}
where the first inequality results from Lemma~\ref{lem:policyave}, and the second inequality results from Lemma~\ref{lem:mdupdate}, and the last inequality follows from the definition of $\varepsilon_{h}^{\alpha}$, which is defined as the upperbound of
\begin{align*}
    \Big|\big\langle\hatQ_{h}^{\lambda,\alpha}(s_{h}^{\alpha},\cdot,\pi_{t}^{\alpha},\mu_{t}^{\calI})-Q_{h}^{\lambda,\alpha}(s_{h}^{\alpha},\cdot,\pi_{t}^{\alpha},\mu_{t}^{\calI}),\pi_{h}^{*,\alpha}(\cdot\,|\,s_{h}^{\alpha})-\pi_{t,h}^{\alpha}(\cdot\,|\,s_{h}^{\alpha}) \big\rangle\Big|\leq \varepsilon_{h}^{\alpha} \text{ for all }t\in[T].
\end{align*}

Taking expectation with respect to $\mu_{h}^{*,\alpha}$ on the both sides of inequality~\eqref{ieq:3}, we can upper bound the difference between $D(\pi_{t+1}^{\calI})$ and $D(\pi_{t}^{\calI})$ as
\begin{align}
    &D(\pi_{t+1}^{\calI})-D(\pi_{t}^{\calI})\nonumber\\
    &\quad=\int_{0}^{1}\sum_{h=1}^{H}\bbE_{\mu_{h}^{*,\alpha}}\Big[\kl\big(\pi_{h}^{*,\alpha}(\cdot\,|\,s_{h}^{\alpha})\|\pi_{t+1,h}^{\alpha}(\cdot\,|\,s_{h}^{\alpha})\big)-\kl\big(\pi_{h}^{*,\alpha}(\cdot\,|\,s_{h}^{\alpha})\|\pi_{t,h}^{\alpha}(\cdot\,|\,s_{h}^{\alpha})\big)\Big]\rmd\alpha\nonumber\\
    &\quad \leq \eta\int_{0}^{1} J^{\lambda,\alpha}(\pi_{t}^{\alpha},\mu_{t}^{\calI})-J^{\lambda,\alpha}(\pi^{*,\alpha},\mu_{t}^{\calI})\rmd\alpha-\lambda\eta\int_{0}^{1}\sum_{h=1}^{H}\bbE_{\mu_{h}^{*,\alpha}}\Big[\kl\big(\pi_{h}^{*,\alpha}(\cdot\,|\,s_{h}^{\alpha})\|\pi_{t,h}^{\alpha}(\cdot\,|\,s_{h}^{\alpha})\big)\Big]\rmd\alpha\nonumber\\
    &\quad\qquad+\frac{1}{2}\eta^{2}H\bigg(H+\lambda H\log|\calA|+\lambda\log\frac{|\calA|}{\beta}\bigg)^{2}+\frac{\beta}{1-\beta}H+2\eta\int_{0}^{1}\sum_{h=1}^{H}\bbE_{\mu_{h}^{*,\alpha}}[\varepsilon_{h}^{\alpha}]\rmd\alpha\nonumber\\
    &\quad\leq-\lambda\eta D(\pi_{t}^{\calI})+\frac{1}{2}\eta^{2}H\bigg(H+\lambda H\log|\calA|+\lambda\log\frac{|\calA|}{\beta}\bigg)^{2}+\frac{\beta}{1-\beta}H+2\eta\int_{0}^{1}\sum_{h=1}^{H}\bbE_{\mu_{h}^{*,\alpha}}[\varepsilon_{h}^{\alpha}]\rmd\alpha,\label{ieq:4}
\end{align}
where the first equation results from the definition of $H(\cdot)$, the first inequality results from inequality~\eqref{ieq:3} and Lemma~\ref{lem:pdl}, and the last inequality results from Proposition~\ref{prop:mono} and the definition of \ac{ne}. 
The inequality~\eqref{ieq:4} can be reformulated as
\begin{align*}
    D(\pi_{t}^{\calI})&\leq \frac{1}{\lambda\eta}\big(D(\pi_{t}^{\calI})-D(\pi_{t+1}^{\calI})\big)+\frac{\eta}{2\lambda}H\bigg(H+\lambda H\log|\calA|+\lambda\log\frac{|\calA|}{\beta}\bigg)^{2}+\frac{\beta H}{(1-\beta)\lambda\eta}\\*
    &\quad\qquad+\frac{2}{\lambda}\int_{0}^{1}\sum_{h=1}^{H}\bbE_{\mu_{h}^{*,\alpha}}[\varepsilon_{h}^{\alpha}]\rmd\alpha.
\end{align*}
Thus, we have that
\begin{align*}
    \frac{1}{T}\sum_{t=1}^{T}D(\pi_{t}^{\calI})&\leq \frac{1}{T\lambda\eta}D(\pi_{1}^{\calI})+\frac{\eta}{2\lambda}H\bigg(H+\lambda H\log|\calA|+\lambda\log\frac{|\calA|}{\beta}\bigg)^{2}+\frac{\beta H}{(1-\beta)\lambda\eta}\\
    &\quad\qquad+\frac{2}{\lambda}\int_{0}^{1}\sum_{h=1}^{H}\bbE_{\mu_{h}^{*,\alpha}}[\varepsilon_{h}^{\alpha}]\rmd\alpha.
\end{align*}
Take $\eta=O(T^{-1/2})$ and $\beta=O(T^{-1})$, then we have $\beta/(1-\beta)=O(T^{-1})$. Thus, we have
\begin{align*}
    \frac{1}{T}\sum_{t=1}^{T}D(\pi_{t}^{\calI})=O\bigg(\frac{\lambda\log^{2} T}{T^{1/2}}+\frac{\int_{0}^{1}\sum_{h=1}^{H}\bbE_{\mu_{h}^{*,\alpha}}[\varepsilon_{h}^{\alpha}]\rmd\alpha}{\lambda}\bigg).
\end{align*}
The desired result in Theorem~~\ref{thm:monoopt} follows from the convexity of KL divergence. Thus, we conclude the proof of Theorem~\ref{thm:monoopt} by noting that $\varepsilon_{h}^{\alpha}=0$ in this case.
\end{proof}
\section{Proof of Theorem~\ref{thm:optest}}
\begin{proof}[Proof of Theorem~\ref{thm:optest}]
    The proof of Theorem~\ref{thm:optest} mainly involves two steps:
    \begin{itemize}
        \item Derive the performance guarantee for Algorithm~\ref{algo:est}.
        \item Combine the estimation result with the optimization result in Theorem~\ref{thm:monoopt}.
    \end{itemize}
    
    \textbf{Step 1: Derive the performance guarantee for Algorithm~\ref{algo:est}.}
    
    Now we focus on estimating the action-value function of $i^{\rm th}$ player with policy $\pi_{t}^{i/N}$. We first introduce some notations. The nominal action-value and the value functions of the $i^{\rm th}$ player with policy $\pi_{t}^{i/N}$ and underlying distribution flow $\mu_{t}^{\calI}$ are respectively denoted as $Q_{h}^{\lambda,i}(s,a)=Q_{h}^{\lambda,i}(s,a,\pi_{t}^{i/N},\mu_{t}^{\calI})$ and $V_{h}^{\lambda,i}(s)=V_{h}^{\lambda,i}(s,\pi_{t}^{i/N},\mu_{t}^{\calI})$ for $h\in[H]$. In the following, we use $z_{h}^{\alpha}$ to denote the aggregate $z_{h}^{\alpha}(\mu_{t,h}^{\calI},W_{h})$. The distribution of $(s_{h}^{i},a_{h}^{i},r_{h}^{i},s_{h+1}^{i})$ induced by the behavior policy $\pi_{t}^{\rmb,i/N}$ and the distribution flow $\mu_{t}^{\calI}$ as $(s_{h}^{i},a_{h}^{i},r_{h}^{i},s_{h+1}^{i})\sim\mu_{h}^{\rmb,i}\times\pi_{h}^{\rmb,i}\times\delta_{r_{h}}\times P_{h}=\tilde{\rho}_{h}^{\rmb,i}$, and the marginalization of this distribution on $(s_{h}^{i},a_{h}^{i},r_{h}^{i})$ is denoted as $\rho_{h}^{\rmb,i}$. For any function $f:\calS\rightarrow\bbR$, we define an operator $\calP_{h}$ as $(\calP_{h}f)(s,a)=\bbE_{s^{\prime}\sim P_{h}(\cdot|s,a)}[f(s^{\prime})]$. Then we adopt recurrence on the time step $h\in[H]$ to derive the estimation performance guarantee.
    
    For time $h=H$, Line 7 in Algorithm~\ref{algo:est} simplifies to
    \begin{align*}
        \hatQ_{H}^{\lambda,i}=\argmin_{f\in\calF_{h}}\sum_{\tau=1}^{K}\big(f(s_{\tau,h}^{i},a_{\tau,h}^{i})-r_{\tau,h}^{i}\big)^{2}
    \end{align*}
    This corresponds to the classical non-parametric regression problem, Theorem 11.4 of \citet{gyorfi2002distribution} shows that the performance guarantee can be derived as
    \begin{align*}
        \bbE_{\rho_{H}^{\rmb,i}}\Big[\big|Q_{H}^{\lambda,i}(s,a)-\hatQ_{H}^{\lambda,i}(s,a)\big|^{2}\Big]=O\bigg(\frac{B_{H}^{4}}{K}\log\frac{\calN_{\infty}(5B_{H}/K,\calF_{H})}{\delta}\bigg).
    \end{align*}
    
    For a time step $h\in[H-1]$, we have that
    \begin{align}
        &\bbE_{\rho_{h}^{\rmb,i}}\Big[\big|Q_{h}^{\lambda,i}(s,a)-\hatQ_{h}^{\lambda,i}(s,a)\big|^{2}\Big]\nonumber\\
        &\quad\leq 2\bbE_{\rho_{h}^{\rmb,i}}\Big[\big|Q_{h}^{\lambda,i}(s,a)-r_{h}(s,a,z_{h}^{i/N})-(\calP_{h}\hatV_{h+1}^{\lambda,i})(s,a)\big|^{2}\Big]\nonumber\\
        &\quad\qquad+2\bbE_{\rho_{h}^{\rmb,i}}\Big[\big|r_{h}(s,a,z_{h}^{i/N})+(\calP_{h}\hatV_{h+1}^{\lambda,i})(s,a)-\hatQ_{h}^{\lambda,i}(s,a)\big|^{2}\Big].\label{ieq:8}
    \end{align}
    For the first term in the right-hand side of inequality~\eqref{ieq:8}, we have
    \begin{align}
        &\bbE_{\rho_{h}^{\rmb,i}}\Big[\big|Q_{h}^{\lambda,i}(s,a)-r_{h}(s,a,z_{h}^{i/N})-(\calP_{h}\hatV_{h+1}^{\lambda,i})(s,a)\big|^{2}\Big]\nonumber\\
        &\quad=\bbE_{\rho_{h}^{\rmb,i}}\Big[\big|r_{h}(s,a,z_{h}^{i/N})+(\calT_{h}^{\pi_{t}^{i/N}}Q_{h+1}^{\lambda,i})(s,a)-r_{h}(s,a,z_{h}^{i/N})-(\calT_{h}^{\pi_{t}^{i/N}}\hatQ_{h+1}^{\lambda,i})(s,a)\big|^{2}\Big]\nonumber\\
        &\quad\leq C_{1}^{2}\cdot\bbE_{\rho_{h+1}^{\rmb,i}}\Big[\big|Q_{h+1}^{\lambda,i}(s,a)-\hatQ_{h+1}^{\lambda,i}(s,a)\big|^{2}\Big],\label{ieq:9}
    \end{align}
    where the equality follows from the definition of operator $\calT$, and the inequality follows from Assumption~\ref{assump:concen}. In the follow, we control the second term in the right-hand side of inequality~\eqref{ieq:8}. For any function $g\in\calF_{h+1}$, we define the value function at time $h+1$ for $\pi_{t}^{i/N}$ induced by $g$ as
    \begin{align*}
        V_{g}^{i}(s)=\langle g(s,\cdot),\pi_{h+1}^{i/N}(\cdot\,|\,s)\rangle+\lambda R\big(\pi_{h+1}^{i/N}(\cdot\,|\,s)\big).
    \end{align*}
    Then the value function defined in Line 8 of Algorithm~\ref{algo:est} can be expressed as $\hatV_{h+1}^{\lambda,i}=V_{\hatQ_{h+1}^{\lambda,i}}^{i}$. For any function $g\in\calF_{h+1}$, we define the regression problem and the corresponding estimate
    \begin{align*}
        \hatQ_{g}=\argmin_{f\in\calF_{h}}\sum_{\tau=1}^{K}\big(f(s_{\tau,h}^{i},a_{\tau,h}^{i})-r_{\tau,h}^{i}-V_{g}^{i}(s_{\tau,h+1}^{i})\big)^{2}.
    \end{align*}
    Then $\hatQ_{h}^{\lambda,i}=\hatQ_{\hatQ_{h+1}^{\lambda,i}}$. Thus, the second term on the right-hand side of inequality~\ref{ieq:8} can be bounded as
    \begin{align*}
        &\bbE_{\rho_{h}^{\rmb,i}}\Big[\big|r_{h}(s,a,z_{h}^{i/N})+(\calP_{h}\hatV_{h+1}^{\lambda,i})(s,a)-\hatQ_{h}^{\lambda,i}(s,a)\big|^{2}\Big]\\
        &\quad\leq \sup_{g\in\calF_{h+1}}\bbE_{\rho_{h}^{\rmb,i}}\Big[\big|r_{h}(s,a,z_{h}^{i/N})+(\calP_{h} V_{g}^{i})(s,a)-\hatQ_{g}(s,a)\big|^{2}\Big].
    \end{align*}
    The term inside the supremum can be handled as
    \begin{align*}
        &\bbE_{\rho_{h}^{\rmb,i}}\Big[\big|r_{h}(s,a,z_{h}^{i/N})+(\calP_{h} V_{g}^{i})(s,a)-\hatQ_{g}(s,a)\big|^{2}\Big]\\
        &\quad=\bbE_{\tilde{\rho}_{h}^{\rmb,i}}\Big[\big|r_{h}(s,a,z_{h}^{i/N})+ V_{g}^{i}(s^{\prime})-\hatQ_{g}(s,a)\big|^{2}\Big]\\
        &\quad\qquad-\bbE_{\tilde{\rho}_{h}^{\rmb,i}}\Big[\big|r_{h}(s,a,z_{h}^{i/N})+(\calP_{h} V_{g}^{i})(s,a)-r_{h}(s,a,z_{h}^{i/N})- V_{g}^{i}(s^{\prime})\big|^{2}\Big],
    \end{align*}
    which follows from the basic calculation. For any $g\in\calF_{h+1}$, Assumption~\ref{assump:complete} implies that
    \begin{align*}
        r_{h}(\cdot,\cdot,z_{h}^{i/N})+(\calP_{h}V_{g}^{i})\in\calF_{h}.
    \end{align*}
    Thus, the definition of $\hatQ_{g}$ shows that
    \begin{align*}
        &\sum_{\tau=1}^{K}\big(r_{\tau,h}^{i}+V_{g}^{i}(s_{\tau,h+1}^{i})-\hatQ_{g}(s_{\tau,h}^{i},a_{\tau,h}^{i})\big)^{2}\leq\sum_{\tau=1}^{K}\big(r_{\tau,h}^{i}+V_{g}^{i}(s_{\tau,h+1}^{i})-r_{\tau,h}^{i}-(\calP_{h}V_{g}^{i})(s_{\tau,h}^{i},a_{\tau,h}^{i})\big)^{2}.
    \end{align*}
    Further, we have that
    \begin{align*}
        &\sup_{g\in\calF_{h+1}}\bbE_{\rho_{h}^{\rmb,i}}\Big[\big|r_{h}(s,a,z_{h}^{i/N})+(\calP_{h} V_{g}^{i})(s,a)-\hatQ_{g}(s,a)\big|^{2}\Big]\\
        &\quad\leq \sup_{g\in\calF_{h+1}}\bigg\{\bbE_{\tilde{\rho}_{h}^{\rmb,i}}\Big[\big|r_{h}(s,a,z_{h}^{i/N})+ V_{g}^{i}(s^{\prime})-\hatQ_{g}(s,a)\big|^{2}\Big]-\bbE_{\tilde{\rho}_{h}^{\rmb,i}}\Big[\big|(\calP_{h} V_{g}^{i})(s,a)-V_{g}^{i}(s^{\prime})\big|^{2}\Big]\\
        &\quad\qquad -\frac{2}{K}\Big[\sum_{\tau=1}^{K}\big(r_{\tau,h}^{i}+V_{g}^{i}(s_{\tau,h+1}^{i})-\hatQ_{g}(s_{\tau,h}^{i},a_{\tau,h}^{i})\big)^{2}-\sum_{\tau=1}^{K}\big(V_{g}^{i}(s_{\tau,h+1}^{i})-(\calP_{h}V_{g}^{i})(s_{\tau,h}^{i},a_{\tau,h}^{i})\big)^{2}\Big]\bigg\}\\
        &\quad\leq \sup_{g\in\calF_{h+1},f\in\calF_{h}}\bigg\{\bbE_{\tilde{\rho}_{h}^{\rmb,i}}\Big[\big|r_{h}(s,a,z_{h}^{i/N})+ V_{g}^{i}(s^{\prime})-f(s,a)\big|^{2}\Big]-\bbE_{\tilde{\rho}_{h}^{\rmb,i}}\Big[\big|(\calP_{h} V_{g}^{i})(s,a)-V_{g}^{i}(s^{\prime})\big|^{2}\Big]\\
        &\quad\qquad -\frac{2}{K}\Big[\sum_{\tau=1}^{K}\big(r_{\tau,h}^{i}+V_{g}^{i}(s_{\tau,h+1}^{i})-f(s_{\tau,h}^{i},a_{\tau,h}^{i})\big)^{2}-\sum_{\tau=1}^{K}\big(V_{g}^{i}(s_{\tau,h+1}^{i})-(\calP_{h}V_{g}^{i})(s_{\tau,h}^{i},a_{\tau,h}^{i})\big)^{2}\Big]\bigg\}.
    \end{align*}
    We define that
    \begin{align*}
        e_{g,f}(s,a,s^{\prime})=\big(r_{h}(s,a,z_{h}^{i/N})+V_{g}^{i}(s^{\prime})-f(s,a)\big)^{2}-\big(V_{g}^{i}(s^{\prime})-(\calP_{h}V_{g}^{i})(s,a)\big)^{2}.
    \end{align*}
    Then we have that
    \begin{align*}
        &\bbE_{\rho_{h}^{\rmb,i}}\Big[\big|r_{h}(s,a,z_{h}^{i/N})+(\calP_{h}\hatV_{h+1}^{\lambda,i})(s,a)-\hatQ_{h}^{\lambda,i}(s,a)\big|^{2}\Big]\\
        &\quad\leq \sup_{g\in\calF_{h+1},f\in\calF_{h}} \bigg\{\bbE_{\tilde{\rho}_{h}^{\rmb,i}}\big[e_{g,f}(s,a,s^{\prime})\big]-\frac{2}{K}\sum_{\tau=1}^{K}e_{g,f}(s_{\tau,h}^{i},a_{\tau,h}^{i},s_{\tau,h+1}^{i})\bigg\}.
    \end{align*}
    Define $B_{H}=(1+\lambda\log|\calA|)H$. The bound for the generalization error is as follows.
    \begin{proposition}\label{prop:gene}
        For any $\varepsilon,\gamma,\theta>0$ we have
        \begin{align*}
            &\bbP\bigg(\exists f\in\calF_{h},g\in\calF_{h+1}, \bbE_{\tilde{\rho}_{h}^{\rmb,i}}\big[e_{g,f}(s,a,s^{\prime})\big]-\frac{1}{K}\sum_{\tau=1}^{K}e_{g,f}(s_{\tau,h}^{i},a_{\tau,h}^{i},s_{\tau,h+1}^{i})\\
            &\quad\qquad \geq \varepsilon\Big(\gamma+\theta+\bbE_{\tilde{\rho}_{h}^{\rmb,i}}\big[e_{g,f}(s,a,s^{\prime})\big]\Big)\bigg)\\
            &\quad\leq 12 \calN_{\infty}\bigg(\frac{\varepsilon\theta}{320B_{H}^{3}},\calF_{h}\bigg)\cdot\calN_{\infty}\bigg(\frac{\varepsilon\theta}{320B_{H}^{3}},\calF_{h+1}\bigg)\cdot \exp\bigg(-\frac{\varepsilon^{2}(1-\varepsilon)\gamma K}{280(1+\varepsilon)B_{H}^{4}}\bigg).
        \end{align*}
    \end{proposition}
    \begin{proof}[Proof of Proposition~\ref{prop:gene}]
        See Appendix~\ref{app:gene}.
    \end{proof}
    We take $\varepsilon=1/2$, $\gamma=\theta=t/2$, then we have
    \begin{align*}
        &\bbP\bigg(\sup_{g\in\calF_{h+1},f\in\calF_{h}} \bigg\{\bbE_{\tilde{\rho}_{h}^{\rmb,i}}\big[e_{g,f}(s,a,s^{\prime})\big]-\frac{2}{K}\sum_{\tau=1}^{K}e_{g,f}(s_{\tau,h}^{i},a_{\tau,h}^{i},s_{\tau,h+1}^{i})\bigg\}\geq t\bigg)\\
        &\quad\leq 12 \calN_{\infty}\bigg(\frac{t}{1280B_{H}^{3}},\calF_{h}\bigg)\cdot\calN_{\infty}\bigg(\frac{t}{1280B_{H}^{3}},\calF_{h+1}\bigg)\cdot \exp\bigg(-\frac{t K}{6720B_{H}^{4}}\bigg).
    \end{align*}
    Thus, with probability at least $1-\delta$, we have
    \begin{align*}
        &\bbE_{\rho_{h}^{\rmb,i}}\Big[\big|r_{h}(s,a,z_{h}^{i/N})+(\calP_{h}\hatV_{h+1}^{\lambda,i})(s,a)-\hatQ_{h}^{\lambda,i}(s,a)\big|^{2}\Big]\\
        &\quad=O\bigg(\frac{B_{H}^{4}}{K}\log\frac{\calN_{\infty}(5B_{H}/K,\calF_{h})\calN_{\infty}(5B_{H}/K,\calF_{h+1})}{\delta}\bigg)
    \end{align*}
    Substituting this inequality and inequality~\eqref{ieq:9} into inequality~\eqref{ieq:8}, we have that
    \begin{align*}
        &\bbE_{\rho_{h}^{\rmb,i}}\Big[\big|Q_{h}^{\lambda,i}(s,a)-\hatQ_{h}^{\lambda,i}(s,a)\big|^{2}\Big]\nonumber\\
        &\quad\leq C_{1}^{2}\cdot\bbE_{\rho_{h+1}^{\rmb,i}}\Big[\big|Q_{h+1}^{\lambda,i}(s,a)-\hatQ_{h+1}^{\lambda,i}(s,a)\big|^{2}\Big]\\
        &\quad\qquad +O\bigg(\frac{B_{H}^{4}}{K}\log\frac{\calN_{\infty}(5B_{H}/K,\calF_{h})\calN_{\infty}(5B_{H}/K,\calF_{h+1})}{\delta}\bigg).
    \end{align*}
    Define the maximal covering number $\calN_{\infty}(\delta,\calF_{[H]})=\max_{h\in[H]}\calN_{\infty}(\delta,\calF_{h})$. Then from the union bound, we have that with probability at least $1-\delta$, for any $i\in[N]$, $h\in[H]$
    \begin{align*}
        \bbE_{\rho_{h}^{\rmb,i}}\Big[\big|Q_{h}^{\lambda,i}(s,a)-\hatQ_{h}^{\lambda,i}(s,a)\big|^{2}\Big]=O\bigg(C_{1}^{2}\frac{HB_{H}^{4}}{K}\log\frac{NH\cdot\calN_{\infty}(5B_{H}/K,\calF_{[H]})}{\delta}\bigg).
    \end{align*}
    
    \textbf{Step 2: Combine the estimation result with the optimization result in Theorem~\ref{thm:monoopt}.}
    
    From the proof of Theorem~\ref{thm:monoopt}, we need to bound the term $\varepsilon_{h}^{\alpha}$. We divide the interval $\calI=[0,1]$ into $N$ small intervals $\calI_{i}=((i-1)/N,i/N]$ for $i\in\{2,\cdots,N\}$ and $\calI_{1}=[0,i/N]$. For any $\alpha\in\calI_{i}$, we have that
    \begin{align*}
        &\Big|\big\langle\hatQ_{h}^{\lambda,\alpha}(s_{h}^{\alpha},\cdot,\pi_{t}^{\alpha},\mu_{t}^{\calI})-Q_{h}^{\lambda,\alpha}(s_{h}^{\alpha},\cdot,\pi_{t}^{\alpha},\mu_{t}^{\calI}),\pi_{h}^{*,\alpha}(\cdot\,|\,s_{h}^{\alpha})-\pi_{t,h}^{\alpha}(\cdot\,|\,s_{h}^{\alpha}) \big\rangle\Big|\\
        &\quad\leq \Big|\big\langle\hatQ_{h}^{\lambda,i}(s_{h}^{\alpha},\cdot)-Q_{h}^{\lambda,i}(s_{h}^{\alpha},\cdot),\pi_{h}^{*,\alpha}(\cdot\,|\,s_{h}^{\alpha})-\pi_{t,h}^{\alpha}(\cdot\,|\,s_{h}^{\alpha}) \big\rangle\Big|\\
        &\quad\qquad+\Big|\big\langle Q_{h}^{\lambda,i}(s_{h}^{\alpha},\cdot)-Q_{h}^{\lambda,\alpha}(s_{h}^{\alpha},\cdot,\pi_{t}^{\alpha},\mu_{t}^{\calI}),\pi_{h}^{*,\alpha}(\cdot\,|\,s_{h}^{\alpha})-\pi_{t,h}^{\alpha}(\cdot\,|\,s_{h}^{\alpha}) \big\rangle\Big|\\
        &\quad \leq 2C_{1}\Big|\big\langle\hatQ_{h}^{\lambda,i}(s_{h}^{\alpha},\cdot)-Q_{h}^{\lambda,i}(s_{h}^{\alpha},\cdot),\pi_{t,h}^{\rmb,i}(\cdot\,|\,s_{h}^{\alpha}) \big\rangle\Big|+O\bigg(\frac{\lambda\log T}{N}\bigg),
    \end{align*}
    where the first inequality results from the triangle inequality, the second inequality results from Assumption~\ref{assump:concen}, the Lipschitzness of reward function in Assumption~\ref{assump:lip}, the Lipschiz constant of $R$ for distributions $p\geq \unif(\calA)/T$, and Cauchy–Schwarz inequality, and we omit the Lipschitz constant dependency on $L_{r}$ for ease of notation. Then with probability at least $1-\delta$, the first term on the right-hand side of this inequality can be controlled as
    \begin{align*}
        &\bbE_{\mu_{h}^{*,\alpha}}\Big[\Big|\big\langle\hatQ_{h}^{\lambda,i}(s_{h}^{\alpha},\cdot)-Q_{h}^{\lambda,i}(s_{h}^{\alpha},\cdot),\pi_{t,h}^{\rmb,i}(\cdot\,|\,s_{h}^{\alpha}) \big\rangle\Big|\Big]\\
        &\quad\leq C_{2}\bbE_{\rho_{h}^{\rmb,i}}\Big[\Big|\big\langle\hatQ_{h}^{\lambda,i}(s_{h}^{\alpha},\cdot)-Q_{h}^{\lambda,i}(s_{h}^{\alpha},\cdot),\pi_{t,h}^{\rmb,i}(\cdot\,|\,s_{h}^{\alpha}) \big\rangle\Big|\Big]\\
        &\quad\leq C_{2}\sqrt{\bbE_{\rho_{h}^{\rmb,i}}\Big[\big(\hatQ_{h}^{\lambda,i}(s,a)-Q_{h}^{\lambda,i}(s,a)\big)^{2}\Big]}\\
        &\quad\leq C_{1}C_{2}\frac{\sqrt{H}B_{H}^{2}}{\sqrt{K}}\log\frac{TNH\cdot\calN_{\infty}(5B_{H}/K,\calF_{[H]})}{\delta},
    \end{align*}
    where the first inequality results from Assumption~\ref{assump:concen}, the second inequality results from H\"older inequality, and the last inequality results from Step 1 and the union bound for $t\in[T]$. Thus, we have
    \begin{align*}
        \int_{0}^{1}\sum_{h=1}^{H}\bbE_{\mu_{h}^{*,\alpha}}[\varepsilon_{h}^{\alpha}]\rmd\alpha=O\bigg(C_{1}C_{2}\frac{H^{3/2}B_{H}^{2}}{\sqrt{K}}\log\frac{TNH\cdot\calN_{\infty}(5B_{H}/K,\calF_{[H]})}{\delta}+\frac{\lambda H\log T}{N}\bigg).
    \end{align*}
    Combined with the proof of Theorem~\ref{thm:monoopt}, this concludes the proof of Theorem~\ref{thm:optest}.
\end{proof}
\section{Proof of Proposition~\ref{prop:uniquene}}
\begin{proof}[Proof of Proposition~\ref{prop:uniquene}]
    The existence of the \ac{ne} follows from Theorem~\ref{thm:exist}. Here we only prove that there are at most one \ac{ne}. Suppose there exists two different \ac{ne}s $(\pi^{\calI},\mu^{\calI})$ and $(\tilde{\pi}^{\calI},\tilde{\mu}^{\calI})$. According to the definition of \ac{ne}, we have that
    \begin{align*}
        &\int_{0}^{1}J^{\lambda,\alpha}(\pi^{\alpha},\mu^{\calI})-J^{\lambda,\alpha}(\tilde{\pi}^{\alpha},\mu^{\calI})\rmd\alpha\geq 0\\
        &\int_{0}^{1}J^{\lambda,\alpha}(\tilde{\pi}^{\alpha},\tilde{\mu}^{\calI})-J^{\lambda,\alpha}(\pi^{\alpha},\tilde{\mu}^{\calI})\rmd\alpha\geq 0.
    \end{align*}
    Summing these two inequalities, we have
    \begin{align*}
        \int_{0}^{1}J^{\lambda,\alpha}(\pi^{\alpha},\mu^{\calI})+J^{\lambda,\alpha}(\tilde{\pi}^{\alpha},\tilde{\mu}^{\calI})-J^{\lambda,\alpha}(\tilde{\pi}^{\alpha},\mu^{\calI})-J^{\lambda,\alpha}(\pi^{\alpha},\tilde{\mu}^{\calI})\rmd\alpha\geq 0,
    \end{align*}
    which contradicts the strictly weak monotone condition.
\end{proof}
\section{Proof of Proposition~\ref{prop:mono}}
\begin{proof}[Proof of Proposition~\ref{prop:mono}]
    We first note that
    \begin{align*}
        J^{\lambda,\alpha}(\pi^{\alpha},\mu^{\calI})-J^{\lambda,\alpha}(\pi^{\alpha},\tilde{\mu}^{\calI})&=J^{\alpha}(\pi^{\alpha},\mu^{\calI})-J^{\alpha}(\pi^{\alpha},\tilde{\mu}^{\calI}),\\*
        J^{\lambda,\alpha}(\tilde{\pi}^{\alpha},\tilde{\mu}^{\calI})-J^{\lambda,\alpha}(\tilde{\pi}^{\alpha},\mu^{\calI})&=J^{\alpha}(\tilde{\pi}^{\alpha},\tilde{\mu}^{\calI})-J^{\alpha}(\tilde{\pi}^{\alpha},\mu^{\calI}),
    \end{align*}
    since the transition kernel is independent of the distribution flow, where we denote $J^{\lambda,\alpha}(\pi^{\alpha},\mu^{\calI})$ for $\lambda=0$ as $J^{\alpha}(\pi^{\alpha},\mu^{\calI})$. Thus, the desired inequality is equivalent to
    \begin{align*}
        \int_{0}^{1}J^{\alpha}(\pi^{\alpha},\mu^{\calI})+J^{\alpha}(\tilde{\pi}^{\alpha},\tilde{\mu}^{\calI})-J^{\alpha}(\tilde{\pi}^{\alpha},\mu^{\calI})-J^{\alpha}(\pi^{\alpha},\tilde{\mu}^{\calI})\rmd\alpha\leq 0.
    \end{align*}
    We define $\rho_{h}^{\calI}, \tilde{\rho}_{h}^{\calI}\in\Delta(\calS\times\calA)^{\calI}$ for $h\in[H]$ as $\rho_{h}^{\alpha}(s,a)=\mu_{h}^{\alpha}(s)\pi_{h}^{\alpha}(a\,|\,s)$ and $\tilde{\rho}_{h}^{\alpha}(s,a)=\tilde{\mu}_{h}^{\alpha}(s)\tilde{\pi}_{h}^{\alpha}(a\,|\,s)$ for all $h\in[H]$ and $\alpha\in\calI$. Then the weakly monotone condition implies that
    \begin{align*}
        \int_{\calI}\sum_{a\in\calA}\int_{\calS} \big(\rho_{h}^{\alpha}(s,a)-\tilde{\rho}_{h}^{\alpha}(s,a)\big)\Big(r_{h}\big(s,a,z_{h}^{\alpha}(\mu_{h}^{\calI},W_{h})\big)-r_{h}\big(s,a,z_{h}^{\alpha}(\tilde{\mu}_{h}^{\calI},W_{h})\big)\Big)\rmd s\rmd \alpha \leq 0.
    \end{align*}
    Then we have that
    \begin{align*}
        &\int_{0}^{1}J^{\alpha}(\pi^{\alpha},\mu^{\calI})+J^{\alpha}(\tilde{\pi}^{\alpha},\tilde{\mu}^{\calI})-J^{\alpha}(\tilde{\pi}^{\alpha},\mu^{\calI})-J^{\alpha}(\pi^{\alpha},\tilde{\mu}^{\calI})\rmd\alpha\\
        &\quad=\sum_{h=1}^{H}\int_{\calI}\sum_{a\in\calA}\int_{\calS} \big(\rho_{h}^{\alpha}(s,a)-\tilde{\rho}_{h}^{\alpha}(s,a)\big)\Big(r_{h}\big(s,a,z_{h}^{\alpha}(\mu_{h}^{\calI},W_{h})\big)-r_{h}\big(s,a,z_{h}^{\alpha}(\tilde{\mu}_{h}^{\calI},W_{h})\big)\Big)\rmd s\rmd \alpha\\
        &\quad\leq 0.
    \end{align*}
    Thus, we conclude the proof of Proposition~\ref{prop:mono}.
\end{proof}
\section{Supporting Propositions and Lemmas}
\subsection{Proof of Proposition~\ref{prop:mfgcontract}}\label{app:mfgcontract}
\begin{proof}[Proof of Proposition~\ref{prop:mfgcontract}]
        We first prove that $T_{h}^{\nu}u\in C_{h}(\calS)$ for any $u\in C_{h+1}(\calS)$. By Proposition 7.32 in \cite{bertsekas1996stochastic}, $T_{h}^{\nu}u$ is continuous. The sup-norm of it can be upper-bounded as
        \begin{align*}
            \|T_{h}^{\nu}u\|_{\infty}\leq 1+\lambda \log|\calA|+(H-h)(1+\lambda \log|\calA|)=(H-h+1)(1+\lambda \log|\calA|).
        \end{align*}
        For the second claim, we have that 
        \begin{align*}
            &\|T_{h}^{\nu}u-T_{h}^{\nu}v\|_{\infty}\\
            &\quad=\sup_{s\in\calS}\big|\sup_{p\in\Delta(\calA)} \sum_{a\in\calA}p(a)\barr_{h}(s,a,\bar{\nu}_{h})-\lambda R(p)+\sum_{a\in\calA}\int_{\calS}p(a)\barP_{h}(s^{\prime}\,|\,s,a,\bar{\nu}_{h})u(s^{\prime})\rmd s^{\prime}\\
            &\quad\qquad-\sup_{q\in\Delta(\calA)} \sum_{a\in\calA}q(a)\barr_{h}(s,a,\bar{\nu}_{h})-\lambda R(q)+\sum_{a\in\calA}\int_{\calS}q(a)\barP_{h}(s^{\prime}\,|\,s,a,\bar{\nu}_{h})v(s^{\prime})\rmd s^{\prime}\big|\\
            &\quad\leq\sup_{s\in\calS,p\in\Delta(\calA)}\Big|\sum_{a\in\calA}\int_{\calS}p(a)\barP_{h}(s^{\prime}\,|\,s,a,\bar{\nu}_{h})\big(v(s^{\prime})-u(s^{\prime})\big)\rmd s^{\prime}\Big|\\
            &\quad\leq \|u-v\|_{\infty},
        \end{align*}
        where the first inequality results from that $|\sup_{x\in\calX}f(x)-\sup_{x\in\calX}g(x)|\leq\sup_{x\in\calX}|f(x)-g(x)|$ for any real-valued functions $f,g$ and set $\calX$. Thus, we conclude the proof of Proposition~\ref{prop:mfgcontract}.
    \end{proof}
\subsection{Proof of Proposition~\ref{prop:closedgraph}}\label{app:closedgraph}
\begin{proof}[Proof of Proposition~\ref{prop:closedgraph}]
        Let $\{(\nu^{(n)},\xi^{(n)})\}_{n\geq 1}\subseteq\Xi\times\Xi$ be a sequence such that $\xi^{(n)}\in\Gamma(\nu^{(n)})$ for all $n\geq 1$ and $(\nu^{(n)},\xi^{(n)})\rightarrow(\nu,\xi)$ as $n\rightarrow\infty$ with respect to the total variation distance for some $(\nu,\xi)\in\Xi\times\Xi$. To prove the graph of $\Gamma$ is closed, we need to prove that $\xi\in\Gamma(\nu)$.
        
        We first prove that $\xi\in A(\nu)$. For any $n\geq 1$ and $h\in[H]$, we have
        \begin{align*}
            \bar{\xi}_{h+1}^{(n)}(s^{\prime})=\sum_{a\in\calA}\int_{\calS}\barP_{h}(\cdot\,|\,s,a,\bar{\nu}_{h}^{(n)})\nu^{(n)}_{h}(\rmd s,a).
        \end{align*}
        Since $\xi^{(n)}\rightarrow\xi$ in total variation, $\xi_{h}^{(n)}\rightarrow\xi_{h}$ weakly. Take any bounded continuous function $g\in C_{b}(\calS)$. Then
        \begin{align}
            \lim_{n\rightarrow\infty}\sum_{a\in\calA}\int_{\calS}\int_{\calS}g(s^{\prime})\barP_{h}(\rmd s^{\prime}\,|\,s,a,\bar{\nu}_{h}^{(n)})\nu_{h}^{(n)}(\rmd s,a)=\sum_{a\in\calA}\int_{\calS}\int_{\calS}g(s^{\prime})\barP_{h}(\rmd s^{\prime}\,|\,s,a,\bar{\nu}_{h})\nu_{h}(\rmd s,a),\label{eq:12}
        \end{align}
        which results from \cite{langen1981convergence}, $\int_{\calS}g(s^{\prime})\barP_{h}(\rmd s^{\prime}\,|\,s,a,\bar{\nu}_{h}^{(n)})$ converges to $\int_{\calS}g(s^{\prime})\barP_{h}(\rmd s^{\prime}\,|\,s,a,\bar{\nu}_{h})$ continuously, and $\nu^{(n)}$ converges to $\nu$. Eqn.~\eqref{eq:12} implies that $\sum_{a\in\calA}\int_{\calS}\barP_{h}(\cdot\,|\,s,a,\bar{\nu}_{h}^{(n)})\nu^{(n)}_{h}(\rmd s,a)$ weakly converges to $\sum_{a\in\calA}\int_{\calS}\barP_{h}(\cdot\,|\,s,a,\bar{\nu}_{h})\nu_{h}(\rmd s,a)$. Thus, we have
        \begin{align*}
            \bar{\xi}_{h+1}(\cdot)=\sum_{a\in\calA}\int_{\calS}\barP_{h}(\cdot\,|\,s,a,\bar{\nu}_{h})\nu_{h}(\rmd s,a).
        \end{align*}
        We then prove that $\xi\in B(\nu)$. Since $\xi^{(n)}\in B(\nu^{(n)})$, there exists sets $A_{h}^{(n)}\subseteq\calS$ for all $n\geq 1$ and $h\in[H]$ such that $\bar{\xi}_{h}^{(n)}(A_{h}^{(n)})=1$, and for any $n\geq 1$, $h\in[H]$ and $s\in A_{h}^{(n)}$, the following equation holds
        \begin{align}
            \sum_{a\in\calA}\xi_{h}^{(n)}(a\,|\,s)\barr_{h}(s,a,\bar{\nu}_{h}^{(n)})-\lambda R\big(\xi_{h}^{(n)}(\cdot\,|\,s)\big)+\sum_{a\in\calA}\int_{\calS}\xi_{h}^{(n)}(a\,|\,s)\barP_{h}(s^{\prime}\,|\,s,a,\bar{\nu}_{h}^{(n)})\barV_{h+1}^{\lambda,\nu^{(n)}}(s^{\prime})\rmd s^{\prime}=\barV_{h}^{\lambda,\nu^{(n)}}(s).\label{eq:14}
        \end{align}
        
        We construct the set $A_{h}=\cap_{k=1}^{\infty}\cup_{n=k}^{\infty}A_{h}^{(n)}$ for all $h\in[H]$. Then the following proposition shows that $\bar{\xi}_{h}(A_{h})=1$ for all $h\in[H]$.
        \begin{proposition}\label{prop:ioevent}
            Suppose that $\{P_{n}\}_{n= 1}^{\infty}$ and $P$ are distributions on a measurable space, and $P_{n}\rightarrow P$ with respect to the total variation distance as $n\rightarrow\infty$. Take any sequence of sets $\{A_{n}\}_{n=1}^{\infty}$ such that $P_{n}(A_{n})=1$. Then we have
            \begin{align*}
                P\big(\cap_{k=1}^{\infty}\cup_{n=k}^{\infty}A_{n}\big)=1.
            \end{align*}
        \end{proposition}
        \begin{proof}[Proof of Proposition~\ref{prop:ioevent}]
                See Appendix~\ref{app:ioevent}.
        \end{proof}
        We are going to prove that for any $s\in A_{h}$, the following equation holds, 
        \begin{align*}
            \sum_{a\in\calA}\xi_{h}(a\,|\,s)\barr_{h}(s,a,\bar{\nu}_{h})-\lambda R\big(\xi_{h}(\cdot\,|\,s)\big)+\sum_{a\in\calA}\int_{\calS}\xi_{h}(a\,|\,s)\barP_{h}(s^{\prime}\,|\,s,a,\bar{\nu}_{h})\barV_{h+1}^{\lambda,\nu}(s^{\prime})\rmd s^{\prime}=\barV_{h}^{\lambda,\nu}(s).
        \end{align*}
        We first show that the optimal value functions $\barV_{h}^{\lambda,\nu^{(n)}}$ converge to $\barV_{h}^{\lambda,\nu}$ continuously.
        \begin{proposition}\label{prop:vfuncconverge}
            Given $\calS$ is compact, if $\nu^{(n)}\rightarrow\nu$ in total variation, we have
            \begin{align*}
                \lim_{n\rightarrow\infty}\sup_{s\in\calS}\big|\barV_{h}^{\lambda,\nu^{(n)}}(s)-\barV_{h}^{\lambda,\nu}(s)\big|=0 \text{ for all }h\in[H].
            \end{align*}
        \end{proposition}
        \begin{proof}[Proof of Proposition~\ref{prop:vfuncconverge}]
                See Appendix~\ref{app:vfuncconverge}.
        \end{proof}
        From the definition of $A_{h}$, for any $s\in A_{h}$, there is a sequence $\{n_{k}\}_{k=1}^{\infty}$ such that $s\in A_{h}^{(n_{k})}$ for all $k\geq 1$. Since $\xi_{h}^{(n_{k})}\rightarrow\xi_{h}$ in total variation, $\xi_{h}^{(n_{k})}(a\,|\,s)\rightarrow\xi_{h}(a\,|\,s)$ for all $a\in\calA$ and $s\in\supp(\bar{\xi}_{h})$. Thus, we have
        \begin{align}
            &\lim_{k\rightarrow\infty}\sum_{a\in\calA}\xi_{h}^{(n)}(a\,|\,s)\barr_{h}(s,a,\bar{\nu}_{h}^{(n)})-\lambda R\big(\xi_{h}^{(n)}(\cdot\,|\,s)\big)+\sum_{a\in\calA}\int_{\calS}\xi_{h}^{(n)}(a\,|\,s)\barP_{h}(s^{\prime}\,|\,s,a,\bar{\nu}_{h}^{(n)})\barV_{h+1}^{\lambda,\nu^{(n)}}(s^{\prime})\rmd s^{\prime}\nonumber\\
            &=\sum_{a\in\calA}\xi_{h}(a\,|\,s)\barr_{h}(s,a,\bar{\nu}_{h})-\lambda R\big(\xi_{h}(\cdot\,|\,s)\big)+\sum_{a\in\calA}\int_{\calS}\xi_{h}(a\,|\,s)\barP_{h}(s^{\prime}\,|\,s,a,\bar{\nu}_{h})\barV_{h+1}^{\lambda,\nu}(s^{\prime})\rmd s^{\prime},\label{eq:15}
        \end{align}
        which results from \cite{langen1981convergence}, Assumption~\ref{assump:mfg} (3) and (4), Proposition~\ref{prop:vfuncconverge} and $\xi_{h}^{(n_{k})}(a\,|\,s)\rightarrow\xi_{h}(a\,|\,s)$. Combining Eqn.~\eqref{eq:15} and that $\barV_{h+1}^{\lambda,\nu^{(n_{k})}}(s)\rightarrow\barV_{h+1}^{\lambda,\nu}(s)$, we prove the Eqn.~\eqref{eq:14}. Thus, we conclude the proof of Proposition~\ref{prop:closedgraph}.
    \end{proof}
\subsection{Proof of Proposition~\ref{prop:ioevent}}\label{app:ioevent}
\begin{proof}[Proof of Proposition~\ref{prop:ioevent}]
        Define the event $G=\cap_{k=1}^{\infty}\cup_{n=k}^{\infty}A_{n}$ and $B_{k}=\cap_{n=k}^{\infty}A_{n}^{\complement}$. Then we have that $G^{\complement}=\cup_{k=1}^{\infty}B_{k}$. Note that $B_{k}\subseteq B_{k+1}$ and the monotone convergence theorem, we have that
        \begin{align}
            P_{n}\big(G^{\complement}\cap A_{n}\big)=\liminf_{k\rightarrow\infty}P_{n}\big(B_{k}\cap A_{n}\big).\label{eq:13}
        \end{align}
        We then have that
        \begin{align}
            1&=\limsup_{n\rightarrow\infty}\liminf_{k\rightarrow\infty} \Big[P_{n}\big(G\cap A_{n}\big)+P_{n}\big(B_{k}\cap A_{n}\big)\Big]\nonumber\\
            &\leq \liminf_{k\rightarrow\infty}\limsup_{n\rightarrow\infty}\Big[P_{n}\big(G\cap A_{n}\big)+P_{n}\big(B_{k}\cap A_{n}\big)\Big],\label{ieq:7}
        \end{align}
        where the equation results from Eqn.~\eqref{eq:13}. For the second term in the right-hand side of inequality~\eqref{ieq:7}, we fix any $k>0$, then for $n>k$, we have that $P_{n}(B_{k}\cap A_{n})=0$ from the definition of $B_{k}$. Thus, we have that
        \begin{align*}
            \limsup_{n\rightarrow\infty}P_{n}\big(B_{k}\cap A_{n}\big)=0,
            \liminf_{k\rightarrow\infty}\limsup_{n\rightarrow\infty}P_{n}\big(G\cap A_{n}\big)=1,
        \end{align*}
        where the second equation results from the first equation and inequality~\eqref{ieq:7}. Since $P_{n}$ converges to $P$ in total variation distance, the weak convergence of $P_{n}$ to $P$ is guaranteed. Portmanteau Theorem shows that
        \begin{align*}
            P(G)\geq \limsup_{n\rightarrow\infty}P_{n}(G)\geq \liminf_{k\rightarrow\infty}\limsup_{n\rightarrow\infty}P_{n}\big(G\cap A_{n}\big)=1.
        \end{align*}
        This concludes the proof of Proposition~\ref{prop:ioevent}.
    \end{proof}
\subsection{Proof of Proposition~\ref{prop:vfuncconverge}}\label{app:vfuncconverge}
\begin{proof}[Proof of Proposition~\ref{prop:vfuncconverge}]
        For ease of notation, we define $T^{(n)}=T^{\nu^{(n)}}$, $T=T^{\nu}$, $\barV^{\lambda,(n)}=\barV^{\lambda,\nu^{(n)}}$, $\barV^{\lambda}=\barV^{\lambda,\nu}$, $u_{0}^{(n)}=u_{0}=0$ and we let
        \begin{align*}
            u_{k+1}^{(n)}=T^{(n)}u_{k}^{(n)}\quad u_{k+1}=Tu_{k}\text{ for }k\geq 1.
        \end{align*}
        From the contraction property in inequality~\eqref{ieq:5}, we have
        \begin{align*}
            \rho(u_{k}^{(n)},\barV^{\lambda,(n)})\leq \sigma^{k}L_{0}\quad \rho(u_{k},\barV^{\lambda)})\leq \sigma^{k}L_{0}.
        \end{align*}
        We then prove that $\lim_{n\rightarrow\infty}\sup_{s\in\calS}|u_{k,h}^{(n)}(s)-u_{k,h}(s)|=0$ for all $h\in[H]$ and $k\geq 0$. We prove this by induction. When $k=0$, $u_{0,h}^{(n)}(s)=u_{0,h}^{(n)}(s)=0$ from definition. Suppose that the claim holds for $k$ and all $h\in[H]$. Consider $k+1$ and any $h\in[H]$, we have that
        \begin{align*}
            &\sup_{s\in\calS}|u_{k+1,h}^{(n)}(s)-u_{k+1,h}(s)|\\
            &\quad= \sup_{s\in\calS}\big|\sup_{p\in\Delta(\calA)} \sum_{a\in\calA}p(a)\barr_{h}(s,a,\bar{\nu}_{h}^{(n)})-\lambda R(p)+\sum_{a\in\calA}\int_{\calS}p(a)\barP_{h}(s^{\prime}\,|\,s,a,\bar{\nu}_{h}^{(n)})u_{k,h+1}^{(n)}(s^{\prime})\rmd s^{\prime}\\
            &\quad\qquad-\sup_{q\in\Delta(\calA)} \sum_{a\in\calA}q(a)\barr_{h}(s,a,\bar{\nu}_{h})-\lambda R(q)+\sum_{a\in\calA}\int_{\calS}q(a)\barP_{h}(s^{\prime}\,|\,s,a,\bar{\nu}_{h})u_{k,h+1}(s^{\prime})\rmd s^{\prime}\big|\\
            &\quad\leq \sup_{s\in\calS,p\in\Delta(\calA)}\bigg|\sum_{a\in\calA}\int_{\calS}p(a)\barP_{h}(s^{\prime}\,|\,s,a,\bar{\nu}_{h}^{(n)})u_{k,h+1}^{(n)}(s^{\prime})\rmd s^{\prime}-\sum_{a\in\calA}\int_{\calS}q(a)\barP_{h}(s^{\prime}\,|\,s,a,\bar{\nu}_{h})u_{k,h+1}(s^{\prime})\rmd s^{\prime}\bigg|\\
            &\quad\qquad+\sup_{s\in\calS,p\in\Delta(\calA)}\Big|\sum_{a\in\calA}p(a)\big(\barr_{h}(s,a,\bar{\nu}_{h}^{(n)})-\barr_{h}(s,a,\bar{\nu}_{h})\big)\Big|.
        \end{align*}
        From Assumption~\ref{assump:mfg} (2), $\barr_{h}(\cdot,\cdot,\bar{\nu}_{h}^{(n)})$ converges to $\barr_{h}(\cdot,\cdot,\bar{\nu}_{h})$ continuously. Also, since continuous function $u_{k,h+1}^{(n)}$ uniformly converges to $u_{k,h+1}$ on compact sets $\calS$ as $n\rightarrow\infty$, we have that $u_{k,h+1}^{(n)}$ converges to $u_{k,h+1}$ continuously. By \cite[Theorem 3.5]{langen1981convergence} and Assumption~\ref{assump:mfg} (4), we have that $\int_{\calS}\barP_{h}(s^{\prime}\,|\,\cdot,\cdot,\bar{\nu}_{h}^{(n)})u_{k,h+1}^{(n)}(s^{\prime})\rmd s^{\prime}$ continuously converges to $\int_{\calS}\barP_{h}(s^{\prime}\,|\,\cdot,\cdot,\bar{\nu}_{h})u_{k,h+1}(s^{\prime})\rmd s^{\prime}$. Since the continuous convergence is equivalent to uniform convergence on compact sets, we have that $\lim_{n\rightarrow\infty}\sup_{s\in\calS}|u_{k+1,h}^{(n)}(s)-u_{k+1,h}^{(n)}(s)|=0$. 
        
        Thus, we have that
        \begin{align}
            &\sup_{s\in\calS}\big|\barV_{h}^{\lambda,\nu^{(n)}}(s)-\barV_{h}^{\lambda,\nu}(s)\big|\nonumber\\
            &\quad\leq \|\barV_{h}^{\lambda,(n)}-u_{k,h}^{(n)}\|_{\infty}+\|u_{k,h}^{(n)}-u_{k,h}\|_{\infty}+\|\barV_{h}^{\lambda}-u_{k,h}\|_{\infty}\nonumber\\
            &\quad\leq \sigma^{h}\rho(\barV^{\lambda,(n)},u_{k}^{(n)})+\|u_{k,h}^{(n)}-u_{k,h}\|_{\infty}+\sigma^{h}\rho(\barV^{\lambda},u_{k})\nonumber\\
            &\quad\leq 2\sigma^{h+k}L_{0}+\|u_{k,h}^{(n)}-u_{k,h}\|_{\infty}.\label{ieq:6}
        \end{align}
        The right-hand side of inequality~\eqref{ieq:6} can be made arbitrarily small by first choosing a large enough $k$ and then let $n\rightarrow\infty$. Thus, we conclude the proof of Proposition~\ref{prop:vfuncconverge}.
    \end{proof}
\subsection{Proof of Proposition~\ref{prop:gene}}\label{app:gene}
\begin{proof}[Proof of Proposition~\ref{prop:gene}]
        The proof of Proposition~\ref{prop:gene} generally follows the proof of Theorem 11.4 in \citet{gyorfi2002distribution}. Here we only specify the different parts. In the following, we bound the $\ell_{1}$ covering number of function class $\{e_{g,f}\,|\,f\in\calF_{h},g\in\calF_{h+1}\}$ on samples $\{(s_{\tau,h}^{i},a_{\tau,h}^{i},s_{\tau,h+1}^{i})\}_{\tau=1}^{K}$. Assume that we have the $\delta$-covers $\calC_{h}^{\delta}$ and $\calC_{h+1}^{\delta}$ of $\calF_{h}$ and $\calF_{h+1}$ with respect to the $\ell_{\infty}$, i.e., for any $f\in\calF_{h}$, there exists $f_{c}\in\calC_{h}^{\delta}$ such that $\|f-f_{c}\|_{\infty}=\sup_{s\in\calS,a\in\calA}|f(s,a)-f_{c}(s,a)|\leq \delta$. Then for any $e_{g,f}$, we can find $f_{c}\in\calC_{h}^{\delta}$ and $g_{c}\in\calC_{h+1}^{\delta}$ such that $\|f-f_{c}\|_{\infty},\|g-g_{c}\|_{\infty}\leq\delta$. The $\ell_{1}$ distance between $e_{g,f}$ and $e_{g_{c},f_{c}}$ on samples $\{(s_{\tau,h}^{i},a_{\tau,h}^{i},s_{\tau,h+1}^{i})\}_{\tau=1}^{K}$ can be bounded as
        \begin{align}
            &\frac{1}{K}\sum_{\tau=1}^{K}\big|e_{g,f}(s_{\tau,h}^{i},a_{\tau,h}^{i},s_{\tau,h+1}^{i})-e_{g_{c},f_{c}}(s_{\tau,h}^{i},a_{\tau,h}^{i},s_{\tau,h+1}^{i})\big|\nonumber\\
            &\quad\leq \frac{1}{K}\sum_{\tau=1}^{K}\bigg\{\Big|\big(r_{h}(s_{\tau,h}^{i},a_{\tau,h}^{i},z_{h}^{i/N})+V_{g}^{i}(s_{\tau,h+1}^{i})-f(s_{\tau,h}^{i},a_{\tau,h}^{i})\big)^{2}\nonumber\\
            &\quad\qquad -\big(r_{h}(s_{\tau,h}^{i},a_{\tau,h}^{i},z_{h}^{i/N})+V_{g_{c}}(s_{\tau,h+1}^{i})-f_{c}(s_{\tau,h}^{i},a_{\tau,h}^{i})\big)^{2}\Big|\nonumber\\
            &\quad\qquad +\Big|\big(V_{g}^{i}(s_{\tau,h+1}^{i})-(\calP_{h}V_{g}^{i})(s_{\tau,h}^{i},a_{\tau,h}^{i})\big)^{2}-\big(V_{g_{c}}(s_{\tau,h+1}^{i})-(\calP_{h}V_{g_{c}})(s_{\tau,h}^{i},a_{\tau,h}^{i})\big)^{2}\Big|\bigg\}.\label{ieq:10}
        \end{align}
        For the first term in the right-hand side of this inequality can be bounded as
        \begin{align*}
            &\frac{1}{K}\sum_{\tau=1}^{K}\Big|\big(r_{h}(s_{\tau,h}^{i},a_{\tau,h}^{i},z_{h}^{i/N})+V_{g}^{i}(s_{\tau,h+1}^{i})-f(s_{\tau,h}^{i},a_{\tau,h}^{i})\big)^{2}\\
            &\quad\qquad -\big(r_{h}(s_{\tau,h}^{i},a_{\tau,h}^{i},z_{h}^{i/N})+V_{g_{c}}(s_{\tau,h+1}^{i})-f_{c}(s_{\tau,h}^{i},a_{\tau,h}^{i})\big)^{2}\Big|\\
            &\quad\leq \frac{2B_{H}}{K}\sum_{\tau=1}^{K}\big|V_{g}^{i}(s_{\tau,h+1}^{i})-V_{g_{c}}(s_{\tau,h+1}^{i})\big|+\big|f(s_{\tau,h}^{i},a_{\tau,h}^{i})-f_{c}(s_{\tau,h}^{i},a_{\tau,h}^{i})\big|\\
            &\quad\leq 4B_{H}\delta,
        \end{align*}
        where the last inequality results from the definition of $f_{c}$ and $g_{c}$. The second term in the right-hand side of inequality~\ref{ieq:10} can be similarly bounded, then we have that
        \begin{align*}
            \frac{1}{K}\sum_{\tau=1}^{K}\big|e_{g,f}(s_{\tau,h}^{i},a_{\tau,h}^{i},s_{\tau,h+1}^{i})-e_{g_{c},f_{c}}(s_{\tau,h}^{i},a_{\tau,h}^{i},s_{\tau,h+1}^{i})\big|\leq 8B_{H}\delta.
        \end{align*}
        The covering number can be correspondingly bounded as
        \begin{align*}
            \calN_{1}\big(\delta,\{e_{g,f}\},\{(s_{\tau,h}^{i},a_{\tau,h}^{i},s_{\tau,h+1}^{i})\}_{\tau=1}^{K}\big)\leq \calN_{\infty}(\delta/8B_{N},\calF_{h})\cdot\calN_{\infty}(\delta/8B_{N},\calF_{h+1}).
        \end{align*}
        Combined with the proof of Theorem 11.4 in \citet{gyorfi2002distribution}, this concludes the proof of Proposition~\ref{prop:gene}.
    \end{proof}

\begin{lemma}\label{lem:policyave}
    For any two distributions $p^{*},p\in\Delta(\calA)$ and $\hatp=(1-\beta)p+\beta\unif(\calA)$ with $\beta\in(0,1)$. Then
    \begin{align*}
        \kl(p^{*}\|\hatp)&\leq\log\frac{|\calA|}{\beta}\\
        \kl(p^{*}\|\hatp)-\kl(p^{*}\|p)&\leq \beta/(1-\beta).
    \end{align*}
\end{lemma}
\begin{proof}[Proof of Lemma~\ref{lem:policyave}]
    \begin{align*}
        \kl(p^{*}\|\hatp)\leq \Big\langle p^{*},\log \frac{p^{*}}{(1-\beta)p+\beta/|\calA|}\Big\rangle\leq \Big\langle p^{*},\log \frac{1}{\beta/|\calA|}\Big\rangle=\log\frac{|\calA|}{\beta}.
    \end{align*}
    Thus, we prove the first inequality. For the second inequality, we have
    \begin{align*}
        \kl(p^{*}\|\hatp)-\kl(p^{*}\|p)=\Big\langle p^{*},\log \frac{p}{(1-\beta)p+\beta/|\calA|}\Big\rangle\leq \Big\langle p^{*},\log \frac{p}{(1-\beta)p}\Big\rangle \leq \Big\langle p^{*},\frac{\beta}{1-\beta}\Big\rangle=\frac{\beta}{1-\beta},
    \end{align*}
    where the second inequality results from $\log(x)\leq x-1$ for $x>0$. Thus, we conclude the proof of Lemma~\ref{lem:policyave}.
    
\end{proof}

\begin{lemma}[Performance Difference Lemma]\label{lem:pdl}
    Given a policy $\pi^{\calI}$ and the corresponding mean-field flow $\mu^{\calI}$, for any player $\alpha\in\calI$ and any policy $\tilde{\pi}^{\alpha}$, we have
    \begin{align*}
        &V_{1}^{\lambda,\alpha}(s,\tilde{\pi}^{\alpha},\mu^{\calI})-V_{1}^{\lambda,\alpha}(s,\pi^{\alpha},\mu^{\calI})+\lambda \bbE_{\tilde{\pi}^{\alpha},\mu^{\calI}}\bigg[\sum_{h=1}^{H}\kl\big(\tilde{\pi}_{h}^{\alpha}(\cdot\,|\,s_{h}^{\alpha})\|\pi_{h}^{\alpha}(\cdot\,|\,s_{h}^{\alpha})\big)\,|\,s_{1}^{\alpha}=s\bigg]\\
        &\quad= \bbE_{\tilde{\pi}^{\alpha},\mu^{\calI}}\bigg[\sum_{h=1}^{H}\big\langle Q_{h}^{\lambda,\alpha}(s_{h}^{\alpha},\cdot,\pi^{\alpha},\mu^{\calI})-\lambda\log\pi_{h}^{\alpha}(\cdot\,|\,s_{h}^{\alpha}),\tilde{\pi}_{h}^{\alpha}(\cdot\,|\,s_{h}^{\alpha})-\pi_{h}^{\alpha}(\cdot\,|\,s_{h}^{\alpha})\big\rangle\,|\,s_{1}^{\alpha}=s\bigg],
    \end{align*}
    where the expectation $\bbE_{\tilde{\pi}^{\alpha},\mu^{\calI}}$ is taken with respect to the randomness in implementing policy $\tilde{\pi}^{\alpha}$ for player $\alpha$ under the \ac{mdp} induced by $\mu^{\calI}$.
\end{lemma}
\begin{proof}[Proof of Lemma~\ref{lem:pdl}]
    From the definition of $V_{1}^{\lambda,\alpha}(s,\tilde{\pi}^{\alpha},\mu^{\calI})$, we have
    \begin{align}
        &V_{1}^{\lambda,\alpha}(s,\tilde{\pi}^{\alpha},\mu^{\calI})\nonumber\\
        &\quad=\bbE_{\tilde{\pi}^{\alpha},\mu^{\calI}}\bigg[\sum_{h=1}^{H}r_{h}(s_{h}^{\alpha},a_{h}^{\alpha},z_{h}^{\alpha})-\lambda\log\tilde{\pi}_{h}^{\alpha}(a_{h}^{\alpha}\,|\,s_{h}^{\alpha})+V_{h}^{\lambda,\alpha}(s_{h}^{\alpha},\pi^{\alpha},\mu^{\calI})-V_{h}^{\lambda,\alpha}(s_{h}^{\alpha},\pi^{\alpha},\mu^{\calI})\,\bigg|\,s_{1}^{\alpha}=s\bigg]\nonumber\\
        &\quad=\bbE_{\tilde{\pi}^{\alpha},\mu^{\calI}}\bigg[\sum_{h=1}^{H}r_{h}(s_{h}^{\alpha},a_{h}^{\alpha},z_{h}^{\alpha})-\lambda\log\tilde{\pi}_{h}^{\alpha}(a_{h}^{\alpha}\,|\,s_{h}^{\alpha})+V_{h+1}^{\lambda,\alpha}(s_{h+1}^{\alpha},\pi^{\alpha},\mu^{\calI})\nonumber\\
        &\quad\qquad-V_{h}^{\lambda,\alpha}(s_{h}^{\alpha},\pi^{\alpha},\mu^{\calI})\,\bigg|\,s_{1}^{\alpha}=s\bigg]+V_{1}^{\lambda,\alpha}(s,\pi^{\alpha},\mu^{\calI}),\label{eq:8}
    \end{align}
    where the second equation results from the rearrangement from the terms. We then focus on a part of the right-hand side of Eqn.~\eqref{eq:8}.
    \begin{align}
        &\bbE_{\tilde{\pi}^{\alpha},\mu^{\calI}}\big[r_{h}(s_{h}^{\alpha},a_{h}^{\alpha},z_{h}^{\alpha})-\lambda\log\tilde{\pi}_{h}^{\alpha}(a_{h}^{\alpha}\,|\,s_{h}^{\alpha})+V_{h+1}^{\lambda,\alpha}(s_{h+1}^{\alpha},\pi^{\alpha},\mu^{\calI})\,|\,s_{1}^{\alpha}=s\big]\nonumber\\
        &\quad=\bbE_{\tilde{\pi}^{\alpha},\mu^{\calI}}\big[r_{h}(s_{h}^{\alpha},a_{h}^{\alpha},z_{h}^{\alpha})+V_{h+1}^{\lambda,\alpha}(s_{h+1}^{\alpha},\pi^{\alpha},\mu^{\calI})\,|\,s_{1}^{\alpha}=s\big]-\lambda\bbE_{\tilde{\pi}^{\alpha},\mu^{\calI}}\Big[ R\big(\tilde{\pi}_{h}^{\alpha}(\cdot\,|\,s_{h}^{\alpha})\big)\,|\,s_{1}^{\alpha}=s\Big]\nonumber\\
        &\quad=\bbE_{\tilde{\pi}^{\alpha},\mu^{\calI}}\Big[\big\langle Q_{h}^{\lambda,\alpha}(s_{h}^{\alpha},\cdot,\pi^{\alpha},\mu^{\calI}),\tilde{\pi}_{h}^{\alpha}(\cdot\,|\,s_{h}^{\alpha})\big\rangle\,|\,s_{1}^{\alpha}=s\Big]-\lambda\bbE_{\tilde{\pi}^{\alpha},\mu^{\calI}}\Big[ R\big(\tilde{\pi}_{h}^{\alpha}(\cdot\,|\,s_{h}^{\alpha})\big)\,|\,s_{1}^{\alpha}=s\Big],\label{eq:9}
    \end{align}
    where $R(\cdot)$ is the negative entropy function, the inner product $\langle\cdot,\cdot\rangle$ is taken with respect to the action space $\calA$, and the second equation results from the definition of $Q_{h}^{\lambda,\alpha}$ and $V_{h+1}^{\lambda,\alpha}$. Substituting Eqn.~\eqref{eq:9} into Eqn.~\eqref{eq:8} and noting the fact that $V_{h}^{\lambda,\alpha}(s_{h}^{\alpha},\pi^{\alpha},\mu^{\calI})=\langle Q_{h}^{\lambda,\alpha}(s_{h}^{\alpha},\cdot,\pi^{\alpha},\mu^{\calI}),\pi_{h}^{\alpha}(\cdot\,|\,s_{h}^{\alpha})\rangle-R(\pi_{h}^{\alpha}(\cdot\,|\,s_{h}^{\alpha}))$, we derive that
    \begin{align*}
        &V_{1}^{\lambda,\alpha}(s,\tilde{\pi}^{\alpha},\mu^{\calI})-V_{1}^{\lambda,\alpha}(s,\pi^{\alpha},\mu^{\calI})\nonumber\\
        &\quad =\bbE_{\tilde{\pi}^{\alpha},\mu^{\calI}}\bigg[\sum_{h=1}^{H}\big\langle Q_{h}^{\lambda,\alpha}(s_{h}^{\alpha},\cdot,\pi^{\alpha},\mu^{\calI}),\tilde{\pi}_{h}^{\alpha}(\cdot\,|\,s_{h}^{\alpha})-\pi_{h}^{\alpha}(\cdot\,|\,s_{h}^{\alpha})\big\rangle\,|\,s_{1}^{\alpha}=s\bigg]\nonumber\\
        &\quad\qquad-\lambda\bbE_{\tilde{\pi}^{\alpha},\mu^{\calI}}\bigg[\sum_{h=1}^{H}R\big(\tilde{\pi}_{h}^{\alpha}(\cdot\,|\,s_{h}^{\alpha})\big)-R\big(\pi_{h}^{\alpha}(\cdot\,|\,s_{h}^{\alpha})\big)\,|\,s_{1}^{\alpha}=s\bigg]\\
        &\quad =\bbE_{\tilde{\pi}^{\alpha},\mu^{\calI}}\bigg[\sum_{h=1}^{H}\big\langle Q_{h}^{\lambda,\alpha}(s_{h}^{\alpha},\cdot,\pi^{\alpha},\mu^{\calI}),\tilde{\pi}_{h}^{\alpha}(\cdot\,|\,s_{h}^{\alpha})-\pi_{h}^{\alpha}(\cdot\,|\,s_{h}^{\alpha})\big\rangle\,|\,s_{1}^{\alpha}=s\bigg]\nonumber\\
        &\quad\qquad -\lambda\bbE_{\tilde{\pi}^{\alpha},\mu^{\calI}}\bigg[\sum_{h=1}^{H}\kl\big(\tilde{\pi}_{h}^{\alpha}(\cdot\,|\,s_{h}^{\alpha})\|\pi_{h}^{\alpha}(\cdot\,|\,s_{h}^{\alpha})\big)+\big\langle \log\pi_{h}^{\alpha}(\cdot\,|\,s_{h}^{\alpha}),\tilde{\pi}_{h}^{\alpha}(\cdot\,|\,s_{h}^{\alpha})-\pi_{h}^{\alpha}(\cdot\,|\,s_{h}^{\alpha}) \big\rangle\,|\,s_{1}^{\alpha}=s\bigg],
    \end{align*}
    where the last equation results from the definition of the negative entropy $R(\cdot)$. This concludes the proof of Lemma~\ref{lem:pdl}.
    
\end{proof}

\begin{lemma}[Lemma 3.3 in~\cite{cai2020provably}]\label{lem:mdupdate}
    For any distribution $p,p^{*}\in\Delta(\calA)$ and any function $g:\calA\rightarrow [0,H]$, it holds for $q\in\Delta(\calA)$ with $q(\cdot)\propto p(\cdot)\exp\big(\alpha g(\cdot)\big)$ that
    \begin{align*}
        \langle g(\cdot),p^{*}(\cdot)-p(\cdot)\rangle \leq \alpha H^{2}/2+\alpha^{-1}\big[\kl(p^{*}\|p)-\kl(p^{*}\|q)\big].
    \end{align*}
\end{lemma}
\iffalse
\begin{lemma}\label{lem:rlip}
    For a finite alphabet $\calX$, define $R$ as the negative entropy function. For two distributions $p,q$ supported on $\calX$, we have that
    \begin{align*}
        |R(p)-R(q)|\leq \max\Big\{\big\|\log(p)\big\|_{\infty},\big\|\log(q)\big\|_{\infty}\Big\}\|p-q\|_{1}.
    \end{align*}
\end{lemma}
\begin{proof}[Proof of Lemma~\ref{lem:rlip}]
    Then we have that
    \begin{align*}
        |R(p)-R(q)|\leq\int_{0}^{1}\Big|\Big\langle\nabla R\big(q+t(p-q)\big),p-q\Big\rangle\Big|\rmd t\leq \|p-q\|_{1}\int_{0}^{1}\Big\|\log\big(q+t(p-q)\big)\Big\|_{\infty}\rmd t,
    \end{align*}
    where the first inequality results from the definition of integral and the triangle inequality, and the second inequality results from H\"{o}lder's inequality. The desired result follows from the fact that for $t\in[0,1]$
    \begin{align*}
        \Big\|\log\big(q+t(p-q)\big)\Big\|_{\infty}\leq \max\Big\{\big\|\log(p)\big\|_{\infty},\big\|\log(q)\big\|_{\infty}\Big\}.
    \end{align*}
    Thus, we conclude the proof of Lemma~\ref{lem:rlip}.
\end{proof}

\begin{lemma}[Lemma 3 in~\cite{xie2021learning}]\label{lem:kllip}
    Let $p,q,u\in\Delta(\calX)$ be distributions supported on a finite set $\calX$. If $p(x)\geq \alpha_{1}$, $q(x)\geq\alpha_{1}$, and $u(x)\geq \alpha_{2}$ for all $x\in\calX$. Then 
    \begin{align*}
        \kl(p\|u)-\kl(q\|u)\leq \bigg(1+\log\frac{1}{\min\{\alpha_{1},\alpha_{2}\}}\bigg)\|p-q\|_{1}
    \end{align*}
\end{lemma}
\fi

\end{document}